\documentclass[11pt,a4paper,final]{iopart}

\usepackage{iopams}  

\usepackage[breaklinks=true,colorlinks=true,linkcolor=blue,urlcolor=blue,citecolor=blue]{hyperref}

\expandafter\let\csname equation*\endcsname\relax
\expandafter\let\csname endequation*\endcsname\relax
\usepackage{amsmath}

\usepackage{graphicx}
\usepackage{amssymb,amsfonts}
\usepackage{amsmath,mathrsfs}
\usepackage{color}
\usepackage{xcolor}

\definecolor{John}{rgb}{0.0, 0.44, 0.0}
\definecolor{Franco}{rgb}{0.1, 0.0, 0.9}

\newtheorem{theorem}{Theorem}

\newtheorem{corollary}[theorem]{Corollary}

\newtheorem{definition}[theorem]{Definition}
\newtheorem{example}[theorem]{Example}

\newtheorem{lemma}[theorem]{Lemma}

\newtheorem{proposition}[theorem]{Proposition}
\newtheorem{remark}[theorem]{Remark}

\newenvironment{proof}[1][Proof]{\textbf{#1.} }{\ \rule{0.5em}{0.5em}}

\newcommand{\Eye}{\hbox{1\kern-4truept 1}}

\begin{document}

\title[]{Mathematical Models of Markovian Dephasing}
\author{Franco Fagnola$^1$, John E.~Gough$^2$, Hendra I.~Nurdin$^3$, \\ and Lorenza Viola$^4$}

\address{$^1$ Dipartimento di Matematica, Politecnico di Milano, Piazza L. da Vinci 32, 
20133 Milano, Italy}

\address{$^2$ Department of Physics, Aberystwyth University, SY23 3BZ, Wales, 
UK}

\address{$^3$ School of Electrical Engineering and Telecommunications, University of  New South Wales, 
Sydney NSW 2052, Australia }

\address{$^4$ Department of Physics and Astronomy, Dartmouth College, 6127 Wilder 
Laboratory, Hanover, NH 03755, USA }

\ead{franco.fagnola@polimi.it, jug@aber.ac.uk, h.nurdin@unsw.edu.au, lorenza.viola@dartmouth.edu}

\date{\today }

\begin{abstract}
We develop a notion of dephasing under the action of a quantum Markov semigroup in terms of convergence of operators to a block-diagonal form determined by irreducible invariant subspaces. If the latter are all one-dimensional, we say the dephasing is {\em maximal}. With this definition, we show that a key necessary requirement on the Lindblad generator is bistochasticity, and focus on characterizing whether a maximally dephasing evolution may be described in terms of a unitary dilation with only classical noise, as opposed to a genuine non-commutative Hudson-Parthasarathy dilation. To this end, we make use of a seminal result of K\"{u}mmerer and Maassen on the class of commutative dilations of quantum Markov semigroups. In particular, we introduce an intrinsic quantity constructed from the generator, the {\em Hamiltonian obstruction}, which vanishes if and only if the latter admits a self-adjoint representation and quantifies the hindrance to having a classical diffusive noise model.
\end{abstract}


\section{Introduction}

The phenomenon of decoherence describes the loss of quantum coherence over time, and the resulting transition from pure quantum states to classical statistical mixtures, due to the interaction of an open system with its environment, which may physically represent unobserved or otherwise uninteresting degrees of freedom, or a measurement apparatus \cite{Decoherence}. In this work, we focus on {\em continuous-time} Markovian quantum dynamics, described by a quantum Markov semigroup (QMS), with the corresponding generator in Gorini-Kossakowski-Sudarshan-Lindblad canonical form \cite{GKS76,Lindblad76}. There has been a renewed interest in understanding the extent to which decoherence may be described purely in terms of random unitary dynamics arising from classical, {\em commutative} noise models (physically akin to fluctuating external fields), with the goals of both shedding light on fundamentally non-classical dynamical aspects and possibly obtaining computationally more tractable models. In fact, the class of QMS generators arising from commutative dilations was completely determined by K\"{u}mmerer and Maassen \cite{KM87} as far back as 1987. 

Our main aim here, in particular, is to discuss the above question for the simplest yet important scenario of decoherence, namely, one where the coherence decay occurs in a dissipationless fashion, without being accompanied by unrecoverable energy exchange with the environment - often referred to as a ``pure $T_2$-process'' in physics terminology or simply as (pure) {\em dephasing}. Since there are several competing mathematical definitions of dephasing in the literature (see for instance \cite{BNII,AFGG12,AFG13}), our first step is to make the notion of Markovian dephasing more precise.  Our formulation is closest to the one in \cite{BNII}; in the case of {\em maximal} dephasing, it leads to the concept of a {\em stable basis} that recovers the ``pointer basis'' introduced by Zurek \cite{Zurek} and also embodies the simplest information-preserving structure \cite{IPS}. For more recent representative contributions on decoherence through random unitary models see for instance \cite{Abel_Marquardt08,Classical,Strunz,Nori}.

We leverage as a main tool the theory of unitary dilations of a QMS, developed in the context of quantum stochastic calculus by Hudson and Parthasarathy \cite{HP84,Partha92}. Specifically, Hudson and Parthasarathy gave an explicit unitary dilation theory using Fock-space-based environments for QMSs. Their quantum stochastic calculus is based on analogue of the It\={o} calculus for integrals with respect to creation, annihilation and gauge processes, and contains classical situations as a special, commutative case. In our context, the relevant question becomes to characterize the dephasing QMSs that are ``truly quantum,'' in the sense that they actually need the full Hudson-Parthasarathy theory: more precisely, those QMSs that {\em cannot} be described as a unitary dilation using {\em only} classical, commutative noise processes. This is where the K\"{u}mmerer and Maassen Theorem \cite{KM87} enters. They studied QMSs on finite-dimensional Hilbert spaces that admit a dilation to a unitary stochastic evolution with classical noise (referred to as ``essentially commutative Markov dilations'') and gave a characterization of the Lindblad generator of such semigroups (\cite[Theorem 1.1.1]{KM87}). These turned out to be the semigroups driven by classical noises that are diffusive (in the form of Wiener processes), or of the jump type (in the form of Poisson processes), or a combination of such noises. Thus, the problem of characterizing the type of decoherence that may ensue from classical noise ultimately comes down to studying the K\"{u}mmerer and Maassen class.  

We begin our analysis by introducing the required background on QMSs and quantum stochastic differential equations (QSDEs) and by discussing some paradigmatic low-dimensional examples (Sec.\S\ref{sec:backg}).  In Sec.\S\ref{sec:deph} we make our notions of dephasing and maximal dephasing mathematically precise (Definition \ref{def:dephasing}) and characterize, in particular, maximal dephasing QMSs as being {\em diagonal} in the stable basis (Theorem \ref{thm:dephasing}). In Sec.\S\ref{sec:obs} we introduce the concept of {\em Hamiltonian obstruction} associated to a dephasing QMS, and show that its vanishing is equivalent to the existence of a representation of the QMS generator involving only {\em self-adjoint} coupling operators (Theorems \ref{thm:obstruct_normal} \& \ref{thm:no_obstr_iff_selfadjoint}). In Sec.\S\ref{sec:class} we bring those tools to bear on the problem of characterizing essentially commutative dilations of maximally dephasing QMS (Theorem \ref{thm:ess_class}). As a main result, we find that vanishing of the Hamiltonian obstruction is necessary and sufficient for a {\em diffusive} classical dilation to exist, whereas a non-zero obstruction may still be compatible with the existence of a classical dilation that involves Poisson noise processes.

\section{Background}
\label{sec:backg}

For convenience, we take the Hilbert space of the system of interest to be $\mathfrak{h} = \mathbb{C}^N$. The Heisenberg picture form of a QMS consists of a family $\Phi =\{ \Phi_t :t \ge 0\}$ of completely positive maps which are conservative, namely,
$\Phi_t (\Eye_N) =\Eye_N$, $\forall t \ge0$. The standard representation of the
generator \cite{Lindblad76} reads (in units $\hbar=1$):
\begin{eqnarray}
\mathscr{L} X = \frac{1}{2} \sum_{k=1}^d [ L^\ast_k , X] L_k  +  \frac{1}{2} \sum_{k=1}^d  L^\ast_k [ X, L_k] - i [X,H],
\label{eq:GKSL}
\end{eqnarray}
where the operator $H$ is self-adjoint. We will restrict to the special case where the index, $k$, ranges over a finite set, say, $k\in \left\{ 1,\ldots ,d\right\} $: the coupling (or Lindblad) operators, $\mathbf{L}=\left\{ L_{k}\right\} $, are bounded by assumption of a finite-dimensional Hilbert space. The Schr\"{o}dinger-picture version consists of the semigroup of dual maps, $\Phi_t^\star$, and the density matrix evolves as $\rho_t = \Phi_t^\star (\rho_0 )$\footnote{Since in this work we consider exclusively finite-dimensional Hilbert spaces and $\ast$-algebras of matrices, we identify density operators $\rho$ with states on the $\ast$-algebra, by letting $\mathbb{E}[\cdot]={\rm tr}\{\rho\, \cdot\}$. In the general infinite-dimensional case, one would  work with von Neumann algebras and density operators would correspond to normal (i.e., weak $\ast$-continuous) states of the von Neumann algebra.}. This leads to the QMS master equation $\dot {\rho}(t) = \mathscr{L}^\star (\rho(t))$, where the dual generator is 
\begin{eqnarray}
\mathscr{L}^\star \rho =   \sum_{k=1}^d L_k  \rho L_k^\ast  
- \frac{1}{2} \sum_{k=1}^d  \big( L^\ast_kL_k  \rho+ \rho L^\ast_k L_k \big)
+i [ \rho ,H].
\label{eq:GKSL^star}
\end{eqnarray}

It is well known that the Heisenberg representation of the generator given in (\ref{eq:GKSL}) is not unique \cite{Lindblad76} (see also Theorem \ref{thm:equiv} below regarding the degree of freedom in choosing the operators $L_k$ and $H$). However, once fixed, the Schr\"{o}dinger version will inherit the representation (\ref{eq:GKSL^star}) by duality.

\begin{definition}
The representation (\ref{eq:GKSL}) is \textbf{minimal} if the number $d$ is minimal, in which case
it is referred to as the \textbf{rank} of the QMS.
The dual representation (\ref{eq:GKSL^star}) is minimal whenever it is dual to a minimal (\ref{eq:GKSL}).
\end{definition}

\noindent
Note that if the representation (\ref{eq:GKSL}) is minimal, then $\Eye_N,L_1,\cdots,L_d$ are linearly independent \cite[Theorem 30.16]{Partha92}. A key concept in discussing Markovian dynamics is Lindblad's definition of {\em dissipator} \cite{Lindblad76}, namely:
\begin{eqnarray}
\mathscr{D}_{\mathscr{L}} (X,Y ) 
\triangleq \mathscr{L}(X Y) -\mathscr{L} (X  ) Y - X  \mathscr{L} (Y).
\end{eqnarray}
Lindblad showed that the generator of a completely positive semigroup must satisfy the dissipativity property
$\mathscr{D}_{\mathscr{L}} (X^\ast ,X) \geq 0$ \cite{Lindblad76}. One sees that the dissipator vanishes,
$\mathscr{D}_{\mathscr{L}} (X^\ast,X) =0$ for all $X$, if and only if $\mathscr{L}$ is Hamiltonian. Following Lindblad
\cite{Lindblad76}, we say that a generator is \textit{pure} if it takes the form $\mathscr{L} X = \frac{1}{2}  [ L^\ast , X] L
+ \frac{1}{2} L^\ast [ X, L]$, which can be easily seen to correspond to the dissipator
$\mathscr{D}_{\mathscr{L}} (X,X) = [X,L]^\ast [X,L]$. The general form (\ref{eq:GKSL}) is then just a sum of $d$ 
pure generators, plus a Hamiltonian term.

\subsection{Quantum stochastic evolutions} 

Dilations of QMSs were realized through the quantum stochastic calculus of Hudson and Parthasarathy \cite{HP84}, now also often referred to as the {\em SLH formalism} in the context of describing quantum feedback networks \cite{GJ09}. We will work on the joint Hilbert space of the system and field degrees of freedom, $\mathfrak{h} \otimes \mathfrak{F}$, where $\mathfrak{F}$ is a prescribed Fock space on which canonical (bosonic) annihilation and creation processes $B_k(t), B_k(t)^\ast$ (for  $k=1, \ldots, d$) are defined. In this formalism, the evolution with respect to $d$ input processes satisfies unitary quantum stochastic dynamics described by a QSDE of the general form
\begin{eqnarray}
d U_G(t) = dG(t) \,  U_G(t) ,
\qquad U_G (0) = \Eye,
\label{eq:U_QSDE}
\end{eqnarray}
where the \textbf {differential germ} is given by 
\begin{eqnarray}
d  G(t) &=& \sum_{jk} (S_{jk} - \delta_{jk} \,\Eye_N ) d \Lambda_{jk} (t) +\sum_j L_j d B_j (t)^\ast
- \sum_{jk} L_j^\ast S_{jk} dB_k (t) \nonumber\\
 &- & \Big( iH + \frac{1}{2} \sum_k L^\ast_k L_k \Big)dt   .
\label{eq:dG_QSDE}
\end{eqnarray}
Here, the repeated indices are summed from 1 to $d$, $\Lambda_{jk} (t)$ are the exchange processes,
and $\Eye$ (with no subscript) is a shorthand for the identity operator on $\mathfrak{h} \otimes \mathfrak{F}$ \cite{HP84}.

On this joint space, we have the triple $G \sim (\mathbf{S},\mathbf{L},H)$. The $\mathbf{S}$ denotes a $d\times d$ array whose entries, $S_{jk}$, are system operators ($N \times N$ matrices). The $\mathbf{L}$ is a column vector of length $d$ with entries,  $L_j$, that are system operators. Finally, we have the system Hamiltonian $H=H^\ast$. The objects $S_{jk} , L_k ,H$ are operators on the $N$-dimensional system space, $\mathfrak{h}$, and a $d$-dimensional \textbf{multiplicity space} $\mathfrak{K}$ is also associated to the input noise fields. Taking $\Omega$ to be the \textbf{Fock vacuum state}, we obtain a QMS $\Phi_t$ with generator $\mathscr{L}$ as in
(\ref{eq:GKSL}) by the prescription
\begin{eqnarray}
\langle u , \, \Phi_t (X) v \rangle \equiv \langle u \otimes \Omega , \, U(t)^\ast [X \otimes \Eye] U(t) \, v \otimes \Omega \rangle ,
\label{eq:Phi_G}
\end{eqnarray}
for each bounded system operator $X$, where $\langle \,,\,\rangle$ denotes inner product in the appropriate Hilbert space.

\smallskip

It is noteworthy that the scattering matrix $\mathbf{S}$ does {\em not} appear in the Lindbladian (\ref{eq:GKSL}), only the coupling operators and the Hamiltonian. In fact, we have the following:

\begin{proposition}
\label{prop:S}
For $G \sim (\mathbf{S},\mathbf{L},H)$ the generating data for a unitary quantum stochastic evolution $U_G (t)$ as in (\ref{eq:U_QSDE}),
let $\Phi_{G , t}$ denote the corresponding QMS and $\mathscr{L}_G$ the associated Lindbladian. Then
\begin{eqnarray}
U_{(\mathbf{S},\mathbf{L},H)} (t) \, | v \otimes \Omega \rangle = 
U_{(\Eye,\mathbf{L},H)} (t) \, | v \otimes \Omega \rangle , \quad \forall v \in \mathfrak{h}.
\label{eq:U=U}
\end{eqnarray}
Moreover, $ \mathscr{L}_{(\mathbf{S},\mathbf{L},H)} = \mathscr{L}_{(\Eye,\mathbf{L},H)} $ which we will denote as $  \mathscr{L}_{(\mathbf{L},H)}$ for simplicity.
\end{proposition}

\begin{proof}
As the (future-pointing) It\={o} increments $dB_j (t)$ and $d \Lambda_{jk} (t)$ annihilate the (future factor) of the vacuum vector $\vert \Omega \rangle $, it follows from (\ref{eq:U_QSDE}) that
\begin{eqnarray*}
dU_G (t) \,  | v \otimes \Omega \rangle = \bigg[  \sum_j L_j d B_j (t)^\ast
- \Big( iH + \frac{1}{2} \sum_k L^\ast_k L_k  \Big) \,dt \bigg] \,  | v \otimes \Omega \rangle ,
\end{eqnarray*}
which does not depend on $\mathbf{S}$. By the uniqueness of the quantum stochastic process \cite{HP84}, we deduce (\ref{eq:U=U}).
\end{proof}

In what follows, two composition rules will be relevant for combining SLH triples of individual components 
\cite{GJ09}. Let $G\sim \left( \mathbf{S},\mathbf{L},H\right) $ and $G^{\prime }\sim \left( \mathbf{S}^{\prime
},\mathbf{L}^{\prime },H^{\prime }\right) $ be SLH triples with the same system space
and multiplicity space. The {\bf series product} is given by 
\begin{eqnarray}
\hspace*{-5mm}
G \vartriangleleft G^\prime = 
\left( \mathbf{S}^{\prime},\mathbf{L}^{\prime }, H^{\prime }\right) \vartriangleleft\left( \mathbf{S},\mathbf{L},H\right)
\sim \Big( \mathbf{S}^{\prime } \mathbf{S},\, \mathbf{S}^\prime \mathbf{L} + \mathbf{L}^{\prime },\, H^{\prime } +
\text{Im} \left\{ \mathbf{L}^{\prime \ast} \mathbf{S}^\prime \mathbf{L} \right\} \Big) .\quad 
\label{series}
\end{eqnarray}
Likewise, the {\bf concatenation product} is given by 
\begin{eqnarray}
\hspace*{-5mm}
G \boxplus G^\prime = \left( \mathbf{S}^{\prime},\mathbf{L}^{\prime },H^{\prime }\right) \boxplus \left( \mathbf{S},\mathbf{L},H\right)
\sim \bigg(
\bigg[
\begin{array}{cc}
	\mathbf{S} & 0 \\
	0 & \mathbf{S}^\prime
\end{array}\bigg],
\bigg[
\begin{array}{l}
	\mathbf{L}  \\
  \mathbf{L}^\prime 
\end{array}\bigg] ,
H +H^\prime \bigg) .
\label{concat}
\end{eqnarray}

\subsection{Bistochastic quantum Markov semigroups}
\label{sec:bistochastic}
An important class of QMS arises by demanding that the dual also defines a valid QMS:

\begin{definition}
A QMS $\{ \Phi_t : t \geq 0\}$ is \textbf{bistochastic} if its Schr\"{o}dinger dual $ \{\Phi_t^\star : t \ge 0 \}$
is also a QMS, in particular $\Phi^\star_t (\Eye_N ) =\Eye_N$.
\end{definition}

\begin{proposition}
The QMS corresponding to $G\sim (\mathbf{S},\mathbf{L},H)$ is bistochastic if and only if
\begin{eqnarray}
\sum_{k=1}^d L^\ast_k L_k = \sum_{k=1}^d L_k L_k^\ast  .
\label{eq:LL}
\end{eqnarray}
\end{proposition}
\begin{proof}
If the QMS is bistochastic, then $\mathscr{L}_{(\mathbf{L},H)}^\star (\Eye_N)=0$ which implies (\ref{eq:LL}). Conversely, if (\ref{eq:LL}) holds, then $\mathscr{L}_{(\mathbf{L},H)} = \mathscr{L}_{(\mathbf{L}^\ast , H)}$, where $\mathbf{L}^\ast$ means the collection of operators $L_k^\ast$.
\end{proof}

\smallskip

Bistochasticity is therefore synonymous with the \textbf{unital} property, which means that $\mathscr{L}_{(\mathbf{L},H)}^\star (\Eye_N)=0$. In the case of a finite-dimensional Hilbert space  as we have assumed, it follows that the maximally mixed state, $\rho_{\text{max}} = \frac{1}{N} \Eye_N$ is invariant under the Schr\"{o}dinger dual semigroup. A complete characterization of bistochastic generators for the qubit case ($N=2$) is given in \cite{Bacon01}. If we fix a density matrix $\rho_0$ and define $\rho_t = \Phi_t^\star (\rho_0 )$ to be the Schr\"{o}dinger evolution of the state at time $t$, then the purity at time $t$ is defined as $p_t \triangleq {\rm tr} \{ \rho_t ^2 \}$. It is known that the purity decreases monotonically for $\text{dim} \, \mathfrak{h} < \infty$ if and only if the QMS is bistochastic \cite{LSA06} (in the infinite-dimensional case bistochasticity is sufficient though not necessary \cite{LSA06}).

\medskip

Some special cases where condition (\ref{eq:LL}) is satisfied are the following:
\begin{itemize}
\item \textbf{Self-duality (up to a Hamiltonian term):} this occurs when $L_k = L^\ast_k$ for each $k$, that is, 
all the couplings operators are self-adjoint
(note that the dual of $\mathscr{L}_{(\mathbf{L},H)}$ is $\mathscr{L}_{(\mathbf{L},-H)}$ in this case).

\item \textbf{Normal operator dissipation:} this occurs when $L^\ast_k L_k =  L_k L_k^\ast$ for each $k$, that is, 
all the coupling operators are normal.
\end{itemize}

\medskip

For dimension $N\geq 3$, it is known that the above self-duality and normality conditions are sufficient 
but {\em not} necessary for the corresponding QMS to be unital \cite{Bacon01}.

\begin{definition}
A triple $G\sim \left( \mathbf{S},\mathbf{L},H\right) $, with system space $\mathfrak{h}$ and
multiplicity space $\mathfrak{K}$, is said to be \textbf{minimal} if there is no triple $G^{\prime }$ with
the same system space $\mathfrak{h}$ and multiplicity space $\mathfrak{K}^{\prime }$ of lower
dimension such that $\mathscr{L}_{G}=\mathscr{L}_{G^{\prime }}$.
\end{definition}

\noindent 
Alternatively, we say that a representation $\mathscr{L}=\mathscr{L}_{G}$ of
a Lindbladian is minimal if $G$ is minimal: that is, we realize $\mathscr{L}$ through
an SLH model using as few noise channels as possible. Generally, there is no 
physical requirement for an actual model set-up to be minimal. The use of this notion 
is purely for mathematical convenience. The minimality condition can be restated as 
follows: $G\sim \left( \mathbf{S},\mathbf{L},H\right) $ is minimal if and only if the 
set $\left\{\Eye_d,L_{k}:k\right\} $ is linearly independent. (This means that if $
c_{0}\Eye_d+\sum_{k}c_{k}L_{k}=0$, with complex coefficients satisfying $
|c_{0}|^{2}+\sum_{k}|c_{k}|^{2}<\infty $, then $c_{0}=0$ and $c_{k}=0$ for
each $k$.)

\begin{definition}
Two SLH triples $G\sim \left( \mathbf{S},\mathbf{L},H\right) $ and $G^{\prime }\sim \left( \mathbf{S}^{\prime
},\mathbf{L}^{\prime },H^{\prime }\right) $ with the same system space
and multiplicity space are \textbf{Euclidean equivalent} if their series product 
\begin{eqnarray*}
 G^{\prime }=G_{\text{scalar}}\vartriangleleft G,
\end{eqnarray*}
where $G_{\text{scalar}}\sim \left( \mathbf{T},\mathbf{\beta} ,e\Eye_d\right) $, with $\mathbf{T}=\left[
T_{jk}\Eye_d\right] ,\mathbf{\beta} =\left[ \beta _{k}\Eye_d \right] $ and the $T_{jk},\beta
_{k} $ and $e$ complex scalars ($\mathbf{T}$ unitary and $e$ real).
\end{definition}

In terms of the actual coefficients, we have $\mathbf{S}^{\prime }=\mathbf{TS}$, $\mathbf{L}^{\prime
}=\mathbf{T \, L}+\mathbf{\beta} $, and $H^{\prime }=H+e+\text{Im}\left\{ \mathbf{\beta} ^{\ast }\mathbf{T}\mathbf{L}\right\} $
or, explicitly,
\begin{eqnarray}
S_{jk}^{\prime } &=&\sum_{l}T_{jl}S_{lk},  \notag \\
L_{j}^{\prime } &=&\sum_{l}T_{j l}L_{l}+\beta _{j}\Eye_N,  \notag \\
H^{\prime } &=&H+e\Eye_N +\frac{1}{2i}\sum_{jk}\left\{ \beta _{j}^{\ast
}T_{jk}L_{k}-L_{j}^{\ast }T_{kj}^{\ast }\beta _{k}\right\} .
\label{eq:transf_SLH}
\end{eqnarray}
The above transformation properties recover the known conditions for invariance of the Lindbladian under a change in representation (sometimes also referred to as ``gauge freedom'' in the literature), $\mathscr{L}_{(\mathbf{L},H)} =
\mathscr{L}_{(\mathbf{L}',H')}$ \cite{Lindblad76,AFG13,TicozziQIC,Peter2016}.  In particular, the complex damping operator 
\begin{eqnarray}
K \triangleq -\frac{1}{2} \sum_k L^\ast _k L_k -iH, 
\label{eq:K_def}
\end{eqnarray} 
transforms as
\begin{eqnarray}
K^\prime = K -\sum_{jk} \beta_j^\ast T_{jk} L_k -\Big( \frac{1}{2} \sum_k
|\beta_k |^2 +i e\Big) \, \Eye_N .
\label{eq:transf_K}
\end{eqnarray}
\noindent 
The following result is proved in \cite{Partha92}:

\begin{theorem}[Parthasarathy \cite{Partha92} Thm. 30.16]
\label{thm:equiv}
Let $\mathscr{L}$ be a Lindbladian on the space of bounded linear operators on ${\mathfrak{h}}$, $\mathfrak{B}(\mathfrak{h})$, and let $\mathscr{L}=\mathscr{L}_{G}$ be a minimal SLH representation. Then all other minimal representations are Euclidean equivalent.
\end{theorem}

We remark that while several results concerning QMSs  can be formulated in terms of a representative
SLH triple, $G$, these results must then be {\em covariant} under transformation of $G$ to a Euclidean
equivalent one. One way of narrowing down the possible equivalence class is to specify the average for
a fixed state:

\begin{definition}
Let $\mathbb{E}[\cdot]= {\rm tr}\{ \rho \, \cdot\}$ be a state corresponding to a density operator $\rho$.
Then $G \sim(\mathbf{S},\mathbf{L},H)$ is \textbf{centered with respect to} $\mathbb{E}$ if $\mathbb{E}[H]=0$ and $
\mathbb{E}[L_k ] = 0$ for all $k$.
\end{definition}

\noindent
Clearly, we can always center $H$ and all the $L_k$ using an appropriate Euclidean transformation.

\subsection{Illustrative examples}

We illustrate the concepts introduced so far by revisiting some paradigmatic examples. 

\subsubsection{Dephasing (phase damping).}
\label{pd}
With $N=2$, take $\mathfrak{B}(\mathfrak{h})=M_{2}$, the space of $2\times 2$ complex matrices, and consider the $d=1$
input model $G\sim \left(  {1\kern-4truept 1}_2, \sqrt{\gamma}\sigma _{z},0 \right)$. Then the Lindbladian is
\begin{eqnarray}
\mathscr{L}_{G}\left( X\right) =\gamma \left( \sigma _{z}X\sigma
_{z}-X\right) , \quad \gamma >0.
\label{eq:qubit_dephase_Lindbladian}
\end{eqnarray}
In this case, $\mathscr{L}^{\star }= \mathscr{L}$ so the QMS is the
same as its dual, and thus automatically bistochastic. The constants of the motion are those operators 
commuting with $\sigma _{z}$, and these are precisely the operators of the form
$\alpha \Eye_2+\beta \sigma _{z}$ for complex numbers $\alpha ,\beta $. As the QMS
is self-dual, the stationary states must have this form too, so we find the
family $\mathscr{E}=\{  \frac{1}{2}\Eye_2 + \frac{1}{2}\lambda \sigma _{z}:\lambda \in \mathbb{
R},|\lambda |\leq 1\}$. As is well known, the master equation, $\dot{\rho}\left(
t\right) =\mathscr{L}^{\star }\left( \rho \left( t\right) \right) $ can be
solved explicitly and, subject to the initial condition $\rho \left(
0\right) =\bigg[\!
\begin{array}{cc}
\rho _{11}( 0) & \rho _{10}( 0) \\
\rho _{00}( 0) & \rho _{00}( 0)
\end{array}\!
\bigg] $, we have
\begin{eqnarray}
\rho \left( t\right) =\bigg[\!
\begin{array}{cc}
\rho _{11}\left( 0\right) & e^{-\gamma t}\rho _{10}\left( 0\right) \\
e^{-\gamma t}\rho _{00}\left( 0\right) & \rho _{00}\left( 0\right)
\end{array}\!
\bigg]  \;\rightarrow  \; \bigg[\!
\begin{array}{cc}
\rho _{11}\left( 0\right) & 0 \\
0 & \rho _{00}\left( 0\right)
\end{array} \!
\bigg] \in {\mathscr E},
\label{eq:pure_dephase_rho}
\end{eqnarray}
in the asymptotic long-time limit.
The limit therefore depends on the initial state: the diagonal terms, representing
populations, are unchanged, whereas the off-diagonal coherence terms vanish.
Recall that a state $\rho$ is {\bf faithful} if, whenever $X\geq 0$ and ${\rm tr} \{\rho X\} =0$, then we 
must have $X=0$. The family $\mathscr{E}$ has the property that all elements except  
$|\lambda|=1$ are faithful states.  To see this, note that $\rho _{0}=|e_{0}\rangle \langle e_{0}|$ and
$\rho _{1}=|e_{1}\rangle \langle e_{1}|$ are extreme elements in ${\mathscr E}$, and 
every other element is a convex combination of these two. It then follows that
if $X=\left[
\begin{array}{cc}
x_{11} & x_{10} \\
x_{01} & x_{00}
\end{array}
\right] \geq 0$ has vanishing expectation for both $\rho_0$ and $\rho_1$, then $
x_{11}=x_{00}=0$ and positivity of $X$ then requires that
$x_{01}=x_{10}^{\ast }=0$ too.

\smallskip

\begin{remark}
In more general dephasing situations, we have a complete orthonormal basis,
$\left\{ |e_{k}\rangle \right\} _{k}$, for which each pure state
$|e_{k}\rangle \langle e_{k}|$ is stationary, whereas $\Phi _{t}^{\star }\left(
|e_{j}\rangle \langle e_{k}|\right) $ vanishes for large $t$ when $j\neq k$.
The set $\left\{ |e_{k}\rangle \langle e_{k}|\right\} _{k}$ forms a 
family of stationary states, and the convex combinations of these
states yield all the invariant states.
\end{remark}

\subsubsection{Depolarization.}
\label{depol}
By still working with $N=2$, a $d=3$ input model may be constructed by letting
$G\sim \left( {1\kern-4truept 1}_2,(\sqrt{\gamma_x }\sigma _{x}, \sqrt{\gamma_x }\sigma_y, \sqrt{\gamma_x }\sigma_z),0 \right)$.
Accordingly, 
\begin{eqnarray*}
\mathscr{L}_{G}\left( X\right) = \sum_{u=x,y,z} \gamma_u \left( \sigma _{u}X\sigma _{u}- X \right) ,\quad \gamma_u >0, 
\end{eqnarray*}
which corresponds to a sum of pure generators, each implementing a phase damping process with strength $\gamma_u$ along
the $u$-th direction. Clearly, the resulting QMS is still self-dual, hence bistochastic.  Unlike for dephasing, however, the only
stationary state $\rho_\infty$ is now the fully mixed density operator, and no rank-one projector of the form $|e_k\rangle\langle e_k|$
exists, that is invariant under the dynamics.

\subsubsection{Relaxation (amplitude damping).} 
\label{ad}
Again, take $\mathfrak{B}(\mathfrak{h})=M_{2}$, but now let the triple $G\sim
\left(  {1\kern-4truept 1}_2,\sqrt{\gamma }\sigma _{-},0\right) $. Then the Lindbladian is
\begin{eqnarray*}
\mathscr{L}_{G}\left( X\right) =\gamma \Big( \sigma _{+}X\sigma _{-}-\frac{1
}{2}\sigma _{+}\sigma _{-}X-\frac{1}{2}X\sigma _{+}\sigma _{-}\Big) ,\quad \gamma >0, 
\end{eqnarray*}
or, equivalently,
\begin{eqnarray*}
\mathscr{L}_{G}\bigg[
\begin{array}{cc}
x_{11} & x_{10} \\
x_{01} & x_{00}
\end{array}
\bigg] =-\gamma \bigg[
\begin{array}{cc}
x_{11}-x_{00} & \frac{1}{2}x_{10} \\
\frac{1}{2}x_{01} & 0
\end{array}
\bigg] .
\end{eqnarray*}
Since $\sigma_+\sigma_- \neq \sigma_-\sigma_+$, the QMS is not bistochastic. Evidently there are no faithful stationary states
other than multiples of the identity. The master equation $ \dot \rho = \mathscr{L}_G^\star \rho$ can be solved explicitly, yielding 
\begin{eqnarray*}
\rho \left( t\right) =\bigg[\!
\begin{array}{cc}
e^{-\gamma t}\rho _{11}\left( 0\right) & e^{-\gamma t/2}\rho _{10}\left(
0\right) \\
e^{-\gamma t/2}\rho _{00}\left( 0\right) & 1-e^{-\gamma t}\rho _{11}\left(
0\right)
\end{array}\!
\bigg] \;\rightarrow \; \bigg[\!
\begin{array}{cc}
0 & 0 \\
0 & 1
\end{array}\!
\bigg],
\end{eqnarray*}
for large times.
Accordingly, there is a unique stationary state, $\rho_{\infty} \equiv |e_{0}\rangle \langle e_{0}|$,
which is {\em pure} and therefore not faithful.

\subsubsection{Relaxation to a pure state and decay.}
The above example provides the textbook example of a QMS that admits a pure state $\rho_\infty$ as its unique stationary state, and can, as such, model physical processes such as purification or ground-state cooling. In such a case, the subspace of the system's Hilbert space that is associated to {\em non-decaying} components is one-dimensional, with a rank-one orthogonal projector $P$ obeying $\rho_\infty = P \rho_\infty P$, and a corresponding $(N-1)$-dimensional {\em decaying} subspace associated to $Q =\Eye_N-P$ \cite{BNII,FR02,DFSU}. In the general case where the steady-state manifold is not one-dimensional, we may require the orthogonal projection $P$ to additionally obey ${\rm tr}(P)=\text{max}_{\rho_\infty}\{\text{rank}(\rho_\infty)\}$, so that $Q=\Eye_N-P$ gives the maximal orthogonal projection for which $\lim_{t \to \infty} Q\, \Phi_t^\star (\rho) \,Q  =0$, for all initial density operators $\rho$ (see also \cite{DFSU,Albert} for a recent characterization of the properties of the generator based on a block-decomposition into decaying and non-decaying components). Accordingly, no decaying subspace exists ($P=\Eye_N$) for the dephasing and depolarizing Lindbladians described in \S \ref{pd}-\S\ref{depol}.

Conditions under which a QMS may admit a unique pure stationary state have been extensively investigated in the mathematical-physics literature \cite{Frig78,Frigerio85,FR01}, and have received recent attention in connection to dissipative quantum state stabilization, see for instance \cite{TicozziTAC,TicozziQIC,Peter2016}. The following result is worth recalling:

\begin{theorem}[Frigerio \protect\cite{Frigerio85} Thm. 3.2]
\label{thm:Frigerio85}
Suppose that $\mathbb{E}_0$ is a stationary pure
state of a QMS with generator $\mathscr{L}$, say $\mathbb{E}_0 [X] = \langle e_0 , X \, e_0
\rangle$ for a unit vector $\vert e_0\rangle$. Then the generator may be
written in the form $\mathscr{L} = \mathscr{L}_G$, where
\begin{eqnarray}
K \vert e_0 \rangle =0, \qquad L_k \vert e_0 \rangle =0, \quad \forall k, 
\label{eq:Conditions_Frigerio}
\end{eqnarray}
where $K=-\frac{1}{2} \sum_k L^\ast _k L_k -iH$ is the complex damping operator defined in (\ref{eq:K_def}).
\end{theorem}
\begin{proof}
We must have
\begin{eqnarray}
\langle e_0 | \mathscr{L} (X) \, e_0 \rangle =0, \qquad \forall X\in
\mathfrak{B}( \mathfrak{h}).
\label{eq:phi_X}
\end{eqnarray}
Setting $X= | e_0 \rangle \langle e_0 | $ in (\ref{eq:phi_X}), we
therefore also have
\begin{eqnarray*}
0= \langle e_0 | \mathscr{L} \big( | e_0 \rangle \langle e_0 |
\big) \, e_0 \rangle= 2 \text{Re} \, \langle e_0 | K \, e_0 \rangle
+\sum_k | \langle e_0 | L_k e_0 \rangle |^2.
\end{eqnarray*}
Without loss of generality, we may assume that $G$ is centered with respect
to $\mathbb{E}_0$, in which case we must have $\text{Re}\{\mathbb{E}_0 [K]\}=0$.
Specifically, we  have $0= \sum_k \mathbb{E}_0 [ L^\ast_k L_k ] = \sum_k
\| L_k e_0 \|^2$, but this requires that $L_k \vert e_0 \rangle =0$
for all $k$.
Thus, we are left with $K \vert e_0 \rangle = -iH \vert e_0 \rangle $.
Let us take $\{ | e_n \rangle : n \geq 0 \}$ to be a complete orthonormal
basis of $\mathfrak{h}$. Then setting $X= | e_n \rangle \langle e_m | $
in (\ref{eq:phi_X}), we have
\begin{eqnarray*}
0= \langle e_0 | \mathscr{L} \big( | e_n \rangle \langle e_m |
\big) = \delta_{0n} \langle e_0 | K^\ast \, e_m \rangle + \langle
e_n | K e_0 \rangle \delta_{m0}.
\end{eqnarray*}
This implies $\langle e_n | K e_0 \rangle =0$ for all $n \neq 0$.
Hence, $e_0$ is an eigenstate of $K$, and consequently of $H$ also. 
By the centering condition, the eigenvalue is zero, as stated.
\end{proof}

\smallskip

The conditions \eqref{eq:Conditions_Frigerio} are stronger than just centering. In particular, they imply that $\langle e_0 , K \, e_0 \rangle =0$, which does not follow from centering alone. The following corollary, an equivalent version of which is also proved in \cite{TicozziTAC} (Proposition 2), makes this explicit:

\begin{corollary}
\label{cor:Frigerio85}
Under the same conditions as in Theorem \ref
{thm:Frigerio85}, every triple $G \sim (\mathbf{S},\mathbf{L},H)$, for which the generator 
$\mathscr{L} = \mathscr{L}_G$, has the property that $|e_0 \rangle$ is an eigenvector
of $K=- \frac{1}{2} \sum_k L_k^\ast L_k - iH$ and each of the $L_k$, for all $k$.
\end{corollary}
\begin{proof}
Let us take $G$ to be the centered SLH triple in Theorem \ref{thm:Frigerio85},
and consider the triple $G^{\prime}$ obtained by the Euclidean triple in
(\ref{eq:transf_SLH}). We have $L_k^\prime \vert e_0 \rangle = \beta_k \,
\vert e_0 \rangle$, and, by (\ref{eq:transf_K}), $K^\prime \vert e_0
\rangle = - \big( \frac{1}{2} \sum_k | \beta_k |^2 + i e \big) \, | e_0
\rangle$. Dropping the primes gives the result.
\end{proof}

\medskip

\noindent
Note that the conditions that $K$ and the $L_k$ have $| e_0 \rangle$ as
an eigenvector are properties which transform covariantly under Euclidean
transformations, but the condition that $H$ has $| e_0 \rangle$ as an
eigenvector does not.

\section{Dephasing quantum Markov semigroups}
\label{sec:deph}

We say that an orthogonal projection $P$ is \textbf{invariant} under a QMS $\Phi $ if $\Phi _t (P)=P$ for all $t \ge 0$. Such a projector is further said to be \textbf{irreducible} if there is no proper sub-projection which is also invariant under the QMS. A family of mutually orthogonal projectors $\{ P_n\}$ on $\mathfrak{h}$ is \textbf{complete} if $P_n P_m = P_n \delta_{n,m}$ and $\sum_n P_n= \Eye_N$.

\begin{definition}
\label{def:dephasing}
Let $P$ and $Q$ be two invariant orthogonal projections that are mutually orthogonal to each other. We say that they \textbf{dephase} under the QMS if
\begin{eqnarray}
\lim_{t \to \infty} \Phi_t ( PXQ) =0 , \qquad \forall X \in \mathfrak{B} ( \mathfrak{h}),
\label{dephPQ}
\end{eqnarray}
where convergence is understood in matrix norm.  A QMS is said to be \textbf{dephasing} if it admits an \textbf{irreducible invariant dephasing family} of projectors, that is, a complete family $\{ P_n\}$ such that each $P_n$ is invariant and irreducible, and each pair $P_n$ and $P_m$, $n \neq m,$ is dephasing as in \eqref{dephPQ}.  The QMS is said to be \textbf{maximally dephasing} if there exists such a family where all projectors are rank one. Equivalently, for a maximally dephasing QMS, there exists an orthonormal basis $\{ |n\rangle \} $, referred to as a \textbf{stable basis}, such that $P_n = | n \rangle \langle n |$ and $\Phi_{t}\left( |n\rangle \langle m|\right) \rightarrow 0$ as $ t\rightarrow \infty $ for $ n\neq m $. 
\end{definition}

Some comments are in order. To begin with, the above dephasing notion was exemplified in \S \ref{pd} for a single-qubit model. Our definition of an irreducible invariant dephasing family coincides with the notion of dephasing introduced by
Baumgartner and Narnhofer in \cite{BNII}, who formulate in the Schr\"{o}dinger picture as
\begin{eqnarray*}
\lim_{t \to \infty} P_m \Phi^\star_t ( \rho ) P_n =0, \qquad \text{whenever } n\neq m.
\end{eqnarray*}
In the case of maximal dephasing, the stable basis recovers the concept of a pointer basis introduced by
Zurek \cite{Zurek}. Indeed, transferring to the Schr\"{o}dinger picture leads to
\begin{eqnarray*}
\Phi_t ^\star (\rho_0 ) \equiv \sum_{n,m} p_{n,m} | n \rangle \langle m |
\to \sum_n p_n | n \rangle \langle n| , \quad p_n = \langle n | \rho | n \rangle, 
\end{eqnarray*}
in the long time limit, for arbitrary initial density matrices $\rho_0$, implying that \emph{classical} information stored in the index $n$ 
is perfectly preserved \cite{IPS}.

In what follows, we shall denote the {\bf commutant} of a family of operators, say, ${\mathcal P}$, by ${\mathcal P}'$, that is, ${\mathcal P}' = \{ X \in \mathfrak{B} ( \mathfrak{h}) \,|\,  [X,P]=0, \;\forall P \in {\mathcal P} \}$.

\begin{theorem}
\label{thm:irred_dephasing}
If a QMS possesses an irreducible invariant dephasing family $\{ P_n \}$, then in any representation $(\mathbf{L},H)$ we have $ H, L_k \in \{ P_n \}^\prime$, for all $k$.
\end{theorem}
\begin{proof}
Suppose that an orthogonal projection $P$ is invariant, then $\mathscr{L}P=0$, and $\mathscr{D}_{\mathscr{L}}(P,P)=\mathscr{L}P-(\mathscr{L}P)P-P(\mathscr{L}P)=0$. As the dissipation, $\mathscr{D}_{\mathscr{L}}(P,P)$, vanishes we conclude that $P$ must commute with each $L_{k}$, and so $0=\mathscr{L}P\equiv -i\left[ P,H\right] $. Therefore $P$ commutes with $H$ and all the $L_{k}$. It follows that if $\left\{ P_{n}\right\} $ is an invariant family then $P_{n}$ commutes with $H$ and $L_{k}$ for all $n$ and $k$.
\end{proof}

Dephasing family of projectors arise naturally if the Hamiltonian and all the noise operators in $G\sim \left( \mathbf{S},\mathbf{L},H\right)$ are diagonal in the same spectral representation:

\begin{lemma}
\label{lem:PXP}
Suppose that we have the spectral decompositions $H \equiv \sum _n \varepsilon _n P_n $
and $L_k = \sum _n \lambda_{k,n} P_n$ for each $k=1,\ldots , d$. Then for any operator $X$,
\begin{eqnarray}
\Phi_t ( P_n XP_m ) = e ^{z_{nm}t} P_n X P_m ,
\label{eq:lem_relax}
\end{eqnarray}
where
\begin{eqnarray}
z_{nm}=\sum_{k}\Big( \lambda _{k,n}^{\ast }\lambda _{k,m}-\frac{1}{2}|\lambda
_{k,n}|^{2}-\frac{1}{2}|\lambda _{k,m}|^{2}\Big) +i (\varepsilon
_{n}-\varepsilon _{m}).
\end{eqnarray}
Moreover, the family $ \{ P_n \}$ is an invariant (not necessarily irreducible) dephasing family if and only if
$\sum_k | \lambda_{k,n} - \lambda_{k,m} |^2 >0$ for all pairs $n \neq m$.
\end{lemma}
\begin{proof}
Under the stated assumptions, the Lindblad generator takes the form
\begin{eqnarray}
\mathscr{L}\left( P _{n}X P _{m}\right) =z_{nm}\, P_n X P _{m}.
\end{eqnarray}
Equation (\ref{eq:lem_relax}) follows automatically.
In particular, $P_n XP_m$ are the eigen-operators of the Linbladian and the $z_{nm}$ are the eigenvalues.
These numbers can be decomposed into real and imaginary parts, namely,
\begin{eqnarray*}
z_{nm}=-\frac{1}{2}\gamma _{nm}-i\omega _{nm} ,
\end{eqnarray*}
where the {\em dephasing rates} and {\em dephasing frequencies} are respectively given by
\begin{eqnarray}
\gamma _{nm} & \triangleq &\frac{1}{2}\sum_{k}\left| \lambda _{k,n}-\lambda
_{k,m}\right| ^{2},
\label{eq:dephase_rates} \\
\omega _{nm} &\triangleq  &\varepsilon _{m}-\varepsilon _{n}+\text{Im}
\sum_{k}(\lambda _{k,n}^{\ast }\lambda _{k,m}).
\label{eq:dephase_frequency}
\end{eqnarray}
We note that $z_{nm}=0$ if $n=m$, hence each $P_n$ is invariant. More generally, $z_{nm}^{\ast }=z_{mn}$,
and $\gamma _{nm}=\gamma _{mn},\omega _{nm}=-\omega _{mn}$.
If $n\neq m$, then $P_n $ and $P_m$ dephase if
and only if the dephasing rate in (\ref{eq:dephase_rates}), and hence $\sum_{k}\left| \lambda _{k,n}-\lambda _{k,m}\right| ^{2}$, is
strictly positive, as claimed.
\end{proof}

\subsection{Comparison with other definitions of dephasing}

An alternative definition of dephasing QMS is introduced by Avron \textit{et al.}, see Proposition 17 of \cite{AFGG12} and \S 2.3 of \cite{AFG13}. Accordingly, a QMS $\mathscr{L} = \mathscr{L}_{H, \mathbf{L}}$ is \textit{dephasing with respect to the Hamiltonian} $H$ if the commutant of $H$, $\{ H \}^\prime$, (strictly) contains the kernel of $\mathscr{L}$. Since in general the latter obeys $\textrm{ker}(\mathscr{L}) \supseteq \{ H, {\bf L}\}'$ (see for instance Lemma 2 in \cite{Peter2016}), and $\{ H \}' \supsetneq \{ H, {\bf L}\}'$,  Avron \textit{et al.}'s definition can be seen  to be equivalent to the requirement that, for each $k$, we have $L_k = f_k (H)$, in the usual spectral sense, for some function $f_k$. This notion of dephasing formalizes the one most commonly employed in physical settings, where a specified Hamiltonian $H$ is taken to define the ``quantization axis'' of the problem and dephasing does not induce transitions in the energy eigenbasis.

Our Theorem \ref{thm:irred_dephasing} requires that $H, L_1, \ldots , L_d \in \{ P_n \}^\prime$, which will typically be a non-commutative set: this leaves open the possibility that $[H,L_k ] \neq 0$ for some $k$. In fact, the Jacobi identity tells us that $[H, L_k ] $ belongs to $\{ P_n \}^\prime$, as do all higher commutators.

\begin{example} 
\label{twoqubits}
Consider a two-qubit system, $N=4$, with two noise inputs, $d=2$, subject to a QMS $\mathscr{L}_{(\mathbf{L},H)}$ involving 
a Heisenberg exchange Hamiltonian and single-qubit independent dephasing channels.  That is,
\begin{eqnarray*}
H= J(\sigma_x \otimes \sigma_x + \sigma_y \otimes \sigma_y + \sigma_z \otimes \sigma_z), \quad J>0, \\
L_1=\sqrt{\gamma_1}\,\sigma_z\otimes \Eye_2, \quad
L_2 = \sqrt{\gamma_2} \,\Eye_2 \otimes \sigma_z, \quad \gamma_i >0.
\end{eqnarray*}
While $H$ belongs to the commutant of the ``collective error algebra'' ${\mathcal A}_z'$ generated by the total angular momentum operator $S_z=\sigma_z \otimes \Eye_2 + \Eye_2\otimes \sigma_z$, we have $[H, L_i]\ne 0$, $i=1,2$. Thus, the QMS does {\em not} induce pure dephasing relative to the eigenbasis of $H$. Nonetheless, $S_z \in \{ H, \mathbf{L} \}^\prime$ and the projectors corresponding to different $S_z$-eigenvalues, $P_1= |e_0\rangle\langle e_0|$, $P_2= |e_1\rangle\langle e_1| + |e_2\rangle\langle e_2|$, $P_3= |e_3\rangle\langle e_3|$, where $\{ |e_j\rangle ,\, j=0,\ldots, 3\}$ denotes the standard computational basis in ${\mathbb C}^4$, form a complete irreducible invariant dephasing family under $ \mathscr{L}_{(\mathbf{L},H)}$. Hence, the latter is dephasing in the sense of Definition \ref{def:dephasing}. It is interesting to contrast this QMS to the permutation-symmetric case where $\gamma_1=\gamma_2=\gamma$ and a single collective noise input is present, $L = \sqrt{\gamma}\, S_z$. Clearly, $[H, L]=0$, hence $ \mathscr{L}_{(L,H)}$ is dephasing according to {\em both} definitions. In the collective case, two-dimensional subspace corresponding to $P_2$ ($S_z=0$) is a ``decoherence-free subspace'', with the action of $H$ on computational basis states implementing encoded single-qubit gates \cite{NJP2002}.
\end{example}

\medskip

Yet another definition of dephasing has been put forward, more recently, by Burgarth \textit{et al.} \cite{Burgarth_Decay}. For the case where the Lindbladian is pure, namely, if $\mathscr{L} = \mathscr{L}_{(L,0)} $ for a single $L$, the QMS is (maximally) dephasing if it admits a stable basis or, equivalently, if $L$ is a normal operator. It is shown that self-duality, $\mathscr{L} = \mathscr{L}^\star$, is a sufficient (though not necessary) condition for $\mathscr{L}_{L,0}$ to be maximally dephasing. For the multi-channel case where $\mathscr{L} = \sum_k \mathscr{L}_{(L_k,0)}$, the authors further show that self-duality implies the condition given in Eq. \eqref{eq:LL}, 
which in turns means that we may write $\mathscr{L} \equiv \sum_k \mathscr{L}_{(X_k,0)} + \sum_k \mathscr{L}_{(Y_k,0)}$, where $X_k = \frac{1}{2} ( L_k + L_k^\ast )$, $Y_k = \frac{1}{2i} ( L_k - L_k^\ast )$. Accordingly, self-duality of $\mathscr{L}$ implies that it is a sum of pure Linbladians that are maximally dephasing. It should, however, be emphasized that these pure Lindbladians are not required to be maximally dephasing with respect to the {\em same} stable basis, so their concept of dephasing does not imply a (common) invariant stable basis in general. In this respect, this approach does not differentiate between the two-qubit QMSs considered in Example \ref{twoqubits} (with self-duality holding up to the Hamiltonian contribution), nor between them and depolarization, \S\ref{depol}, for which no common stable basis may be found. 

In the same paper, the property of non-self duality is identified as a necessary condition for \textit{bone fide} ``decay'' (e.g., ground-state relaxation as in \S\ref{ad}). While, based on this approach, the authors conclude that decay {\em cannot} be ascribed to classical noise, as captured by a QSDE driven by a classical stochastic field, some caution is needed in interpreting this conclusion, both because it is tied to the assumed notion of dephasing and because consideration is restricted to just the diffusive Wiener class (see further discussion in \S \ref{sec:class}).  Given our dephasing notion, we will establish that dephasing processes exist, which are (also) {\em not} ascribable to ``classical noises'' and, likewise, that in cases where classical noise leads to dephasing, the generator need {\em not} to be self-dual.

\subsection{Maximal dephasing}

\subsubsection{Characterization of a maximally dephasing QMS.}
While Lemma \ref{lem:PXP} does not cover irreducibility, in the case of maximal dephasing it is easy to supply conditions:

\begin{theorem}
\label{thm:dephasing}
A QMS determined from the triple $G\sim \left( \mathbf{S},\mathbf{L},H\right) $ is maximally dephasing if and only if the operators 
$H$ and the $L_{k}$ are diagonal in the stable basis, and for all pairs $n \neq m$ in $\{ 1,\ldots, N\}$ we have $\langle n | \,  L_{k}|n \rangle  \neq \langle m | \,  L_{k} | m \rangle  $ for at least one $k\in\{ 1 , \ldots, d \}$.
\end{theorem}
\begin{proof}
By Corollary \ref{cor:Frigerio85}, if $G\sim \left( \mathbf{S},\mathbf{L},H\right) $ is a triple representing the QMS, each $|n\rangle $ is an eigenvector of $K=-\left( \frac{1}{2}\sum_{k}L_{k}^{\ast }L_{k}+iH\right)  $ and $L_{k}$ for every $k$. Thus, we may write
\begin{eqnarray}
L_{k}\equiv \sum_{n=1}^N\lambda _{k,n}\,|n\rangle \langle n |,\quad K\equiv \sum_{n=1}^N\kappa _{n}\,|n\rangle \langle n|,
\quad  \lambda _{k,n}, \, \kappa _{n} \in {\mathbb C}, \, \forall n.
\label{eq:L_max_decomp}
\end{eqnarray}
By assumption, we also have
\begin{eqnarray}
H\equiv \sum_{n=1}^N \varepsilon _{n}\,|n\rangle \langle n| , \quad \quad  \varepsilon_n \in {\mathbb R}, \, \forall n,
\label{eq:H_max_decomp}
\end{eqnarray}
whereby it follows that $\kappa _{n}=-\frac{1}{2}\sum_{k}\left| \lambda _{k,n}\right| ^{2}-i\varepsilon _{n}$.
With this prescription, we see that Lemma \ref{lem:PXP} applies and we find $\Phi _{t}\left( |n\rangle \langle m|\right)
=e^{z_{nm}t}|n\rangle \langle m|.$ If $n\neq m$, then $|n\rangle $ and $|m\rangle $ dephase if
and only if  $\sum_{k}\left| \lambda _{k,n}-\lambda _{k,m}\right| ^{2}>0$,
which is equivalent to the stated condition upon noticing that $\lambda _{k,n}=\mathbb{E}_{n}\left[ L_{k}
\right] $, the expectation of $L_k$ for pure state $|n \rangle$.
\end{proof}

\medskip

It is convenient to collect the coefficients $\lambda_{k,n}$ into a matrix
\begin{eqnarray}
\begin{array}{cc}  L_1 & \rightarrow \\ L_2 & \rightarrow \\
L_3 & \rightarrow \\
\vdots & \dots \\
L_d & \rightarrow \end{array}
\left[ \begin{array}{ccccc}
\lambda_{11} &  \lambda_{12} & \lambda_{13}  & \dots & \lambda_{1N}\\
\lambda_{21} &  \lambda_{22} & \lambda_{23}  & \dots & \lambda_{2N}\\
\lambda_{31} &  \lambda_{32} & \lambda_{33}  & \dots & \lambda_{3N}\\
\vdots &  \vdots & \vdots & \vdots & \vdots\\
\lambda_{d1} &  \lambda_{d2} & \lambda_{d3}  & \dots & \lambda_{dN}\\
\end{array}
\right] \equiv F.
\label{eq:F_matrix}
\end{eqnarray}
The $k$th row of $F$ corresponds to the operator $L_k = \sum_{n=1}^N \lambda_{k,n} | n \rangle \langle n |$. 
However, we may
also focus on the column vectors:
\begin{eqnarray}
F \equiv  \left[ | \boldsymbol{\mathbf{ \lambda} }_1 \rangle , \ldots ,
| \boldsymbol{\mathbf{\lambda}}_N \rangle \right] , \qquad
| \boldsymbol{\mathbf{\lambda }}_n  \rangle
\equiv
\left[
\begin{array}{c}
\lambda_{1,n} \\
\vdots \\
\lambda_{d,n}\\
\end{array}
\right] \in \mathbb{C}^ d \equiv \mathfrak{K}.
\end{eqnarray}

For a maximal dephasing QMS, we therefore have that $\{ P_n \}^\prime$ consists of the commutative set of operators diagonal in 
the stable basis. In particular, the relations (\ref{eq:L_max_decomp}) and (\ref{eq:H_max_decomp}) arrived at in Theorem 
\ref{thm:dephasing} imply that the $L_k$ and $H$ may be thought of as functions of a common observable 
$Q= \sum_n q_n |n \rangle \langle n |$. If we may take $Q$ to be $H$,  we recover exactly Avron {\em et al.}'s definition 
in \cite{AFG13}. The maximally dephasing condition from Theorem \ref{thm:dephasing} requires that for all $n \neq m$,
$\lambda_{k,n}  \neq \lambda_{k,m}$, for at least one $k$. In fact, we see that the dephasing damping rates $\gamma_{nm}$ 
are half the length-squared of the vector $| \boldsymbol{\mathbf{\lambda}}_n  \rangle - | \boldsymbol{\mathbf{\lambda}}_m  \rangle$ so the condition that these do not vanish for any $n \neq m$ is just that the set $ \left\{ | \boldsymbol{\mathbf{\lambda}}_1 \rangle , \ldots ,
| \boldsymbol{\mathbf{\lambda}}_N \rangle \right\}$ consists of $N$ distinct (though possibly linearly dependent) vectors.

\subsubsection{Rank of a maximally dephasing QMS.}

Suppose we have a pure QMS, that is, one with rank $d=1$. Is it possible for it to realize a maximally dephasing QMS 
on a system with an $N$-dimensional Hilbert space? In this case the matrix $F$ in (\ref{eq:F_matrix}) is simply
$F= [ \lambda_{11} , \ldots , \lambda_{1N} ]$ and our condition for dephasing is that $| \lambda_{1n} - \lambda_{1m} | \neq 0$ 
for all $n \ne m$, which just means that the complex numbers $\lambda_{11} , \ldots ,  \lambda_{1 N}$ are all distinct. If so, 
the corresponding QMS will be a rank-1 maximally dephasing QMS.

More generally, given the class of maximally dephasing QMSs, we can ask for the maximal rank possible.
To this end, consider a rank-$d$ minimal representation with coupling operators $\{ L_1 , \cdots , L_d \}$.
As the set $\{  \Eye_N, L_1, \cdots, L_d\}$ needs to be linearly independent, the extended matrix
$\tilde F$ defined by
\begin{eqnarray*}
\begin{array}{cc}\Eye_N & \rightarrow \\ L_1 & \rightarrow \\ L_2 & \rightarrow \\
L_3 & \rightarrow \\
\vdots & \dots \\
L_d & \rightarrow \end{array}
\left[ \begin{array}{ccccc}
1 & 1 &  1  & \dots & 1 \\
\lambda_{11} &  \lambda_{12} & \lambda_{13}  & \dots & \lambda_{1N}\\
\lambda_{21} &  \lambda_{22} & \lambda_{23}  & \dots & \lambda_{2N}\\
\lambda_{31} &  \lambda_{32} & \lambda_{33}  & \dots & \lambda_{3N}\\
\vdots &  \vdots & \vdots & \vdots & \vdots\\
\lambda_{d1} &  \lambda_{d2} & \lambda_{d3}  & \dots & \lambda_{dN}\\
\end{array}
\right] \equiv \tilde F
\end{eqnarray*}
must have all its $1+d$ rows linearly independent. 
We therefore must have $d +1\leq N$, so the upper limit on the rank must be $N-1$.

\begin{example}
Consider $N=3$  (a qutrit) with $d=2$ noise inputs determined by
\begin{eqnarray*}
| \boldsymbol{\mathbf{\lambda}}_1 \rangle =
\left[ \begin{array}{c}
1\\
2\\
\end{array}
\right], \qquad
| \boldsymbol{\mathbf{\lambda}}_2 \rangle =
\left[ \begin{array}{c}
2\\
4\\
\end{array}
\right], \qquad
| \boldsymbol{\mathbf{\lambda}}_3 \rangle =
\left[ \begin{array}{c}
1\\
0\\
\end{array}
\right],
\end{eqnarray*}
so that $L_1 = | 1 \rangle \langle 1|  + 2 |2 \rangle \langle 2| + |3 \rangle \langle 3|$,
$L_1 = 2| 1 \rangle \langle 1|  + 4 |2 \rangle \langle 2| $. Here $ L_1, L_2$ and the identity $\Eye_3$ are linearly independent 
and the vectors $ | \boldsymbol{\mathbf{\lambda}}_1 \rangle,| \boldsymbol{\mathbf{\lambda}}_2 \rangle,| \boldsymbol{\mathbf{\lambda}}_3 \rangle$ are distinct, 
though not linearly independent. Indeed, $| \boldsymbol{\mathbf{\lambda}}_2 \rangle= 2 |\boldsymbol{ \mathbf{\lambda}}_1 \rangle$.
In this example we have maximal dephasing, and the largest rank $d=3-1$ possible.
\end{example}

\section{Hamiltonian obstruction} 
\label{sec:obs}

For a maximally dephasing QMS, an essential role in establishing Theorem \ref{thm:dephasing} is played by the
dephasing rates, introduced in  (\ref{eq:dephase_rates}). We now turn our attention to the dephasing frequencies $\omega_{nm}$ in
(\ref{eq:dephase_frequency}). First, we show that the set of frequencies $\left\{ \omega_{nm}:n,m\right\}$ are {\em not}, 
in general, attributable to a Hamiltonian term.

To this end, we note that if $z=x+iy$ and $z^{\prime }=x^{\prime }+iy^{\prime }$
are a pair of complex numbers, then $\text{Im} \, \left\{ z^{\ast }z^{\prime }\right\}
=xy^{\prime }-yx^{\prime }$, which geometrically is the (signed) area of the
parallelogram in the complex plane with vertices $0,z,z^{\prime},z+z^{\prime }$. The $d$ coupling operators,
$L_{k}= \sum_{n}\lambda _{k,n}\,|n\rangle \langle n|$, give rise to $N$ vectors $|\boldsymbol{\mathbf{\lambda}}_n \rangle
\in \mathfrak{K} = \mathbb{C}^{d }$. The real values that enter the definition of $\omega_{nm}$, namely,
\begin{eqnarray}
A_{nm} = \text{Im} \sum_{k=1}^d (\lambda _{k,n}^{\ast }\lambda _{k,m})
\equiv \text{Im} \, \langle \boldsymbol{\mathbf{\lambda}}_n | \boldsymbol{\mathbf{\lambda}}_m \rangle ,
\label{eq:symplectic_area}
\end{eqnarray}
are then a sum of $d $ signed areas. (Note that the last expression in
(\ref{eq:symplectic_area}) is an inner product in $\mathfrak{K} =\mathbb{C}^d$,
not the Hilbert space $\mathfrak{h} =\mathbb{C}^N$ of the system.)
We may think of this as the symplectic area of the corresponding parallelogram in $\mathbb{C}^{d }$.

\begin{remark}
\label{rem:obstruction}Let $\left\{ \omega _{nm}:n,m\right\} $ be the dephasing frequencies appearing in 
(\ref{eq:dephase_frequency})  (i.e., the imaginary parts of the generator's eigenvalues).
Then they generally do not take the form $\omega _{nm}=\omega _{m}-\omega _{n}$ for a fixed set of real numbers 
$\left\{ \omega _{n}:n\right\} $.
To see this, assume we did have the form $\omega _{nm}=\omega _{m}-\omega _{n}$. Then if
$n,m,l$ are distinct, the identity $\Delta _{nml}=\omega_{nm}+\omega _{ml}+\omega _{ln}=0$ must hold,
whereas we obtain
\begin{eqnarray}
\hspace*{-5mm}
\Delta _{nml} =
\sum_{k}{\rm{Im}}\left( \lambda _{k,n}^{\ast }\lambda_{k,m}+\lambda _{k,m}^{\ast }\lambda _{k,l}+\lambda _{k,l}^{\ast }\lambda_{k,n}\right) = A_{nm} + A_{ml} + A_{ln}  ,
\label{eq:Delta}
\end{eqnarray}
which is non-vanishing in general.
\end{remark}

\medskip

The right hand-side in equation (\ref{eq:Delta}) has an intrinsic geometrical meaning: it is the symplectic area of the 
triangle in $\mathbb{C}^{d}$
with vertices at $\lambda _{n},\lambda _{m},\lambda _{l}$. Thus, vanishing of the
{\textbf{Hamiltonian obstruction}}, $\Delta_{nml}$, is a {\em necessary} condition for a set of frequencies
$\omega_{nm}$ to stem from an Hamiltonian term of the form $\sum_{n}\omega_{n}|n\rangle \langle n|$.

\begin{example}
Let us take $\mathfrak{h}=\mathbb{C}^{3}$ and choose $d =3$ coupling
operators, $L_{k}\equiv \sum_{n}\lambda _{k,n}\,|n\rangle \langle n|$, $k=1,2,3$, with
\begin{eqnarray*}
| \boldsymbol{\mathbf{\lambda }}_{1} \rangle =
\left[
\begin{array}{c}
1 \\
0 \\
2i\\
\end{array}
\right] ,\quad
| \boldsymbol{\mathbf{\lambda }}_{2} \rangle =
\left[
\begin{array}{c}
i \\
0 \\
1\\
\end{array}
\right] ,\quad
| \boldsymbol{\mathbf{\lambda }}_{3} \rangle =
\left[
\begin{array}{r}
2 \\
0 \\
-1 \\
\end{array}
\right] ,
\quad
| \boldsymbol{\mathbf{\lambda }}_{0} \rangle =
\left[
\begin{array}{c}
1 \\
1 \\
1 \\
\end{array}
\right] .
\end{eqnarray*}
(We include $| \boldsymbol{\mathbf{\lambda }}_{0} \rangle$ as this allows us to construct the identity operator 
as $\Eye_3=\sum_{n} |n\rangle \langle n|$.)
We readily see that the dephasing rates $\gamma_{12},\gamma_{23},\gamma_{31}$ are all
positive-definite so we have maximal dephasing; however, $\Delta _{123}= -5 \ne 0$.
This example is not minimal as $\big\{ | \boldsymbol{\mathbf{\lambda}}_0 \rangle , | \boldsymbol{\mathbf{\lambda }}_{1} \rangle,
| \boldsymbol{\mathbf{\lambda}}_{2} \rangle , | \boldsymbol{\mathbf{\lambda }}_{3} \rangle \big\} $ is clearly over-complete.
\end{example}

\noindent
An obstruction may also arise for a maximally dephasing QMS with a minimal representation, 
as the next example shows.

\begin{example}
Since operators $L_k$ are identified with vectors, taking account of the identity operator
we need at least 3 operators $L_k$, hence the Hilbert space must have dimension at least $N=4$.
Indeed, if one considers
\begin{eqnarray*}
|\boldsymbol{\mathbf{\lambda}}_0 \rangle =
\left[
\begin{array}{c}
1\\
1\\
1\\
1\\
\end{array}
\right] , \quad
| \boldsymbol{\mathbf{\lambda}}_1 \rangle =
\left[
\begin{array}{r}
1\\
 i\\
-1\\
 -i\\
\end{array}
\right] , \quad
| \boldsymbol{\mathbf{\lambda}}_2 \rangle =
\left[
\begin{array}{r}
1\\
-1\\
i\\
-i\\
\end{array}
\right] , \quad
| \boldsymbol{\mathbf{ \lambda}}_3 \rangle =
\left[
\begin{array}{r}
1\\
-1\\
 -1\\
 1\\
\end{array}
\right],
\end{eqnarray*}
the four vectors $|\boldsymbol{\mathbf{\lambda}}_0\rangle$ (corresponding to $\Eye_4$),
$|\boldsymbol{\mathbf{\lambda}}_1\rangle, |\boldsymbol{\mathbf{\lambda}}_2\rangle, |\boldsymbol{\mathbf{\lambda}}_3\rangle$ are linearly independent and $\Delta_{123}= -2 \neq 0$.
\end{example}

\begin{lemma}
\label{lem:obstruction}
The obstruction $\Delta_{nml}$ in (\ref{eq:Delta}) is unchanged under a Euclidean equivalence 
transformation.
\end{lemma}
\begin{proof}
We take $L_j^\prime = \sum_k T_{jk} L_k + \beta_k$. This implies a relationship of the form
\begin{eqnarray*}
\lambda_{j,n} ^ \prime = \sum_k T_{jk} \lambda _k + \beta_k ,
\end{eqnarray*}
for each projection index $n$. Now $A_{nm} = \text{Im} \sum_{k} (\lambda _{k,n}^{\ast }\lambda _{k,m})$ is anti-symmetric, and we compute
\begin{eqnarray*}
A^\prime_{nm} = \sum_{k}\text{Im}( \lambda^{\prime \, \ast}_{k,n} \lambda^\prime _{k,m} )
=  A_{nm} - \sum_k \text{Im} \big\{ \alpha_k^\ast (\lambda_{k,n} - \lambda_{k,m} ) \big\} ,
\end{eqnarray*}
where $\boldsymbol{\mathbf{\alpha}} = \mathbf{ R} \, \boldsymbol{\mathbf{\beta}}$. Since $ \Delta _ {nml} = A_{nm} + A_{ml} + A_{ln}$, 
we find $\Delta _ {nml}^\prime  = A_{nm}^\prime  + A_{ml}^\prime + A_{ln}^\prime
= \Delta_{nml},$ which completes the proof.
\end{proof}

\medskip

Lemma \ref{lem:obstruction} shows that obstructions are fundamental and cannot be removed by using an equivalent 
Euclidean representation of the Lindbladian. We have, in particular:

\begin{theorem}
\label{thm:obstruct_normal}
Given a maximally dephasing QMS, suppose there is an obstruction (i.e., $ \Delta_{nml} \neq 0$
for some $n \neq m \neq l$). Then it is impossible to represent the generator of the QMS
in a form where all the coupling operators are self-adjoint.
\end{theorem}
\begin{proof}
Let us suppose that $L_k^*=L_k$ for all $k$, then $A_{nm} \equiv 0$ for all $n,m$. In this case, $ \Delta_{nml} = 0$ for all $n \neq m \neq l$. This implies that $\Delta_{nml}$ would vanish identically, so having all $L_k$ self-adjoint leads to zero obstruction. Moreover, by Lemma \ref{lem:obstruction}, any Euclidean equivalent model will also have $\Delta \equiv 0$.
\end{proof}

\begin{remark} Note that for any maximally dephasing QMS that is obstruction-free,
we can find a representation of the generator in which all inner products $\left\langle \boldsymbol{\mathbf{\lambda}}_{n}|
\boldsymbol{\mathbf{\lambda}}_{m}\right\rangle$ are real.
Indeed, multiplying vectors $|\boldsymbol{\mathbf{\lambda}}_{n} \rangle $ for 
$n>1$ by an appropriate phase $\hbox{\rm e}^{\mathrm{i}\theta_j}$,
we can make $\left\langle \boldsymbol{\mathbf{\lambda}}_{ 1}|\boldsymbol{\mathbf{\lambda}}_{n}\right\rangle$ real. Consequently, if the obstruction vanishes,
\begin{eqnarray*}
\hspace*{-5mm}
 A_{nm} =  \mathrm{Im} \left\langle \boldsymbol{\mathbf{\lambda}}_{n}|
\boldsymbol{\mathbf{\lambda}}_{m}\right\rangle =
\mathrm{Im} \left  \langle  \boldsymbol{\mathbf{\lambda}}_{n}| \boldsymbol{\mathbf{\lambda}}_{m}\right\rangle +
\mathrm{Im} \left\langle \boldsymbol{\mathbf{\lambda}}_{m}| \boldsymbol{\mathbf{\lambda}}_{1}\right\rangle
+\mathrm{Im} \left\langle \boldsymbol{\mathbf{\lambda}}_{1}| \boldsymbol{\mathbf{\lambda}}_{n}\right\rangle
=0, \quad \forall n,m.
\end{eqnarray*}
\end{remark}

Exploiting the above remark, we can prove a converse to Theorem \ref{thm:obstruct_normal}
(in the finite-dimensional case we are presently considering):

\begin{theorem}
\label{thm:no_obstr_iff_selfadjoint}
Given a maximally dephasing QMS, suppose there is no obstruction (i.e., $\Delta_{nml} = 0$
for all $n \neq m \neq l$). Then it is possible to represent the generator of the QMS
in a form where all the coupling operators are self-adjoint.
\end{theorem}
\begin{proof}
We will look for a form that is minimal. Again, let $\mathfrak{h}=\mathbb{C}^N$ denote the system space and
$\mathfrak{K}=\mathbb{C}^d$ be the multiplicity space, with $d$ being the rank.
We prove now that, if all inner products
$\left\langle \boldsymbol{\mathbf{\lambda}}_n, \boldsymbol{\mathbf{\lambda}}_{m}\right\rangle$ are real,
we can find a basis of $\mathfrak{K}$ in which all the elements $\lambda_{ij}$ are real.
Recall that the matrix $F$ introduced in (\ref{eq:F_matrix}) will have $d$ linearly independent rows,
and that $d+1\le  N$. The $N$ vectors of $F$, $|\boldsymbol{\mathbf{\lambda}}_{n} \rangle$ in $\mathbb{C}^d$,  are
\begin{eqnarray*}
\left[\begin{array}{c} \lambda_{11} \\ \lambda_{21} \\
\vdots \\ \lambda_{d1} \\ \end{array} \right], \quad
\left[\begin{array}{c} \lambda_{12}  \\ \lambda_{22} \\
\vdots \\ \lambda_{d2}  \\ \end{array} \right], \quad
\left[\begin{array}{c} \lambda_{13}  \\ \lambda_{23} \\
\vdots \\ \lambda_{d3} \\ \end{array} \right], \quad \dots \quad
\left[\begin{array}{c} \lambda_{1N} \\ \lambda_{2N} \\
\vdots \\ \lambda_{dN} \\ \end{array} \right]
\end{eqnarray*}
and $d$ of them are linearly independent. Relabeling coordinates in $\mathbb{C}^d$ we
can always assume that the first $d$ are linearly independent.
Acting with a unitary on $\mathbb{C}^d$, if necessary, we can always suppose that they are
written in the form
\begin{eqnarray*}
\left[\begin{array}{c} r_1  \\ 0 \\ 0 \\ \vdots \\ 0 \\ \end{array} \right], \qquad
\left[\begin{array}{c} z_{12}  \\ z_{22} \\ z_{32}\\ \vdots \\ z_{d2}\\ \end{array} \right] , \qquad
\left[\begin{array}{c} z_{13}  \\ z_{23} \\ z_{33} \\ \vdots \\ z_{d3}\\ \end{array} \right] , \quad \dots \quad ,
\left[\begin{array}{c} z_{1N}  \\ z_{2N} \\ z_{3N} \\ \vdots \\ z_{dN}\\ \end{array} \right] ,
\end{eqnarray*}
with $r_1>0$, $z_{ij}\in\mathbb{C}$, the first $d$ vectors being linearly independent
(and the remaining ones being a linear combination of them).
If their inner products are real, they can be written in the form
\begin{eqnarray*}
\left[\begin{array}{c} r_1  \\ 0 \\ 0 \\ \vdots \\ 0 \\ \end{array} \right], \qquad
\left[\begin{array}{c} r_{2}  \\ z_{22} \\ z_{32}\\ \vdots \\ z_{d2}\\ \end{array} \right] , \qquad
\left[\begin{array}{c} r_{3}  \\ z_{23} \\ z_{33} \\ \vdots \\ z_{d3}\\ \end{array} \right] , \quad \dots \quad,
\left[\begin{array}{c} r_{N}  \\ z_{2N} \\ z_{3N} \\ \vdots \\ z_{dN}\\ \end{array} \right] ,
\end{eqnarray*}
with $r_1>0,r_2,\dots, r_N\in\mathbb{R}$.
Acting with a unitary of the form
\begin{eqnarray*}
U=\left[\begin{array}{ccccc} 1 & 0 & 0 &  \dots & 0\\
0 & \ast & \ast & \dots & \ast \\
0 & \dots & \dots & \dots & \dots \\
0 & \ast & \ast & \dots & \ast  \\
\end{array}\right],
\end{eqnarray*}
we may get the vectors
\begin{eqnarray*}
\left[\begin{array}{c} r_1  \\ 0 \\ 0 \\ \dots \\ 0\end{array} \right], \qquad
\left[\begin{array}{c} r_2  \\ s_2 \\ 0 \\ \dots \\ 0 \end{array} \right], \qquad
\left[\begin{array}{c} r_3  \\ z_{23} \\ z_{33} \\ \dots \\ z_{d3}\end{array} \right], \quad \dots \quad ,
\left[\begin{array}{c} r_N  \\ z_{2N} \\ z_{3N} \\ \dots \\ z_{dN} \end{array} \right],
\end{eqnarray*}
with $r_1>0, s_2>0,\, r_2,\dots, r_N\in\mathbb{R}$. Moreover, since inner products are
real, $z_{2j}\in\mathbb{R}$ for all $j\ge 3$, so that we have
\begin{eqnarray*}
\left[\begin{array}{c} r_{11}  \\ 0 \\ 0 \\ \dots \\ 0\end{array} \right], \qquad
\left[\begin{array}{c} r_{12}  \\ r_{22} \\ 0 \\ \dots \\ 0 \end{array} \right], \qquad
\left[\begin{array}{c} r_{13} \\  r_{23} \\ r_{33} \\ \dots \\ r_{d3}\end{array} \right], \quad \dots,
\left[\begin{array}{c} r_{1N}  \\ r_{2n} \\ r_{3N} \\ \dots \\ r_{dN} \end{array} \right],
\end{eqnarray*}
and $r_{jj}>0$ for $j=1,\dots,d$. Iterating this procedure $d$ times, it is clear that we get $N$ vectors in $\mathbb{C}^d$
with real components $r_{kn}$; additionally, the above algorithm gives us $r_{kn} =0$ for $k>n$.
This shows that, with a unitary transformation on $\mathfrak{K}$, we may represent the generator
with self-adjoint operators $L_k \equiv \sum_{n \ge k} r_{kn} | n \rangle \langle n |$, as stated.
\end{proof}

\smallskip

We may combine the above two theorems into the following characterization:
\begin{corollary}
A maximally dephasing QMS has vanishing Hamiltonian obstruction (i.e., $\Delta_{nml} = 0$ for all $n \neq m \neq l$) if
and only if there exists a representation of the generator in which all the coupling operators are self-adjoint.
\end{corollary}

\section{Classical noise in quantum theory}
\label{sec:class}

We now turn to the main question of determining whether a QMS may admit a classical 
unitary dilation. Specifically, we give the following definition \cite{KM87}:

\begin{definition}
\label{def:ec}
Given a QMS on the space $\mathfrak{h}={\mathbb C}^N$, an {\bf essentially 
commutative (or essentially classical) dilation} is one specified by a dilated operator algebra of the form 
$\mathfrak{B}(\mathfrak{h}) \otimes {\mathscr C}$, where $\mathfrak{B}(\mathfrak{h}) = M_N$ is 
the space of $N \times N$ matrices and ${\mathscr C}$ is a 
commutative von Neumann algebra. 
\end{definition}

\smallskip

Consider three elementary models for unitary evolutions in the presence of classical noise.
The first is deterministic: this is just the usual Schr\"{o}dinger unitary evolution,
\begin{eqnarray}
U^{\mathrm{det}}_H(t) = e^{-i H t} ,
\end{eqnarray}
where the Hamiltonian $H$ is taken to be self-adjoint. The second model is a diffusive one, driven by a Wiener process $W(t)$: namely, we take
\begin{eqnarray}
U^{\mathrm{diff}}_R(t) = e^{-iR \, W(t)} ,
\end{eqnarray}
where $R$ is self-adjoint. From the It\={o} calculus, we obtain the QSDE
\begin{eqnarray}
dU^{\mathrm{diff}}_R(t) = \Big(\!-iR  dW(t)  - \frac{1}{2} R^2 \ dt \Big) \, U^{\mathrm{diff}}_R(t) .
\label{diff}
\end{eqnarray}
Finally, we consider a model determined by a Poisson process $N_\nu (t)$: namely,
\begin{eqnarray}
U^{\mathrm{jump}}_{S, \nu}(t) = S^{N_\nu (t)} ,
\label{eq:U_jump}
\end{eqnarray}
where $S$ is taken to be unitary and $\nu >0$ is the rate of the Poisson process, so
that $ \mathbb{E} [dN_\nu (t) ] = \nu dt$. From the It\={o} calculus,
the corresponding QSDE now reads
\begin{eqnarray}
dU_{S, \nu }^{\mathrm{jump}}(t) =   (S-  \Eye_N ) \, dN_\nu (t)   \,
U_{S, \nu }^{\mathrm{jump}}(t) .
\label{jump}
\end{eqnarray}
Physically, each realization of the evolution described by (\ref{diff}) corresponds to a smooth diffusive trajectory,
whilst in case (\ref{jump}) one is effectively applying a unitary kick at random times determined by the Poisson
point process. Although the majority of work in the physics literature has focused on diffusions, models based on 
telegraph processes have also been considered, especially in the context of solid-state qubits, see e.g. 
\cite{Abel_Marquardt08}.

For each of the above cases, we obtain a QMS, $\Phi_t$, by evolving with the corresponding
$U(t)$ and averaging over the noise, that is:
\begin{eqnarray*}
\Phi_t (X) = \mathbb{E} \big[ U(t)^\ast X U(t) \big] .
\end{eqnarray*}
The respective Lindblad generators read:
\begin{eqnarray}
\mathscr{L}_H^{\mathrm{det}} (X) &\!=\!& -i [X,H ] , 
\label{Ldet}\\
\mathscr{L}_R^{\mathrm{diff}} (X) &\!=\!& - \frac{1}{2} \big[ [X,R ], R \big] ,
\label{Ldiff} \\
\mathscr{L}_{S,\nu} ^{\mathrm{jump}} (X) &\!=\!& \nu \, ( S^\ast X S -X ).
\label{Ljump}
\end{eqnarray}

\noindent  
Note that the deterministic and diffusive generators belong to the closure of the cone generated by the jump generators:
\begin{eqnarray*}
\mathscr{L}_H^{\mathrm{det}} = \frac{\partial}{\partial u} \mathscr{L}_{e^{-iHu}}^{\mathrm{jump}} |_{u=0}
=\lim_{u \to 0^+} \frac{1}{u} \mathscr{L}_{e^{-iHu}}^{\mathrm{jump}}  ,\\
\mathscr{L}_R^{\mathrm{diff}} = \frac{\partial^2}{\partial u^2} \mathscr{L}_{e^{-iRu}}^{\mathrm{jump}} |_{u=0}
=\lim_{u \to 0^+} \frac{1}{2u^2} \Big( \mathscr{L}_{e^{-iRu}}^{\mathrm{jump}}
+\mathscr{L}_{e^{+iRu}}^{\mathrm{jump}} \Big)  .
\end{eqnarray*}
 
In \cite{KM87}, K\"ummerer and Maassen show that every QMS that admits an 
essentially classical dilation as defined above also admits a Kraus representation of the form
\begin{eqnarray*}
\Phi_t (X) \equiv \int_{U(N)} \!\!V^\ast XV \, d \mu_t (V) ,
\end{eqnarray*}
with  $\{ \mu_t \}$ being a convolution semigroup of probability measures on the unitary group $U(N)$.
A well-known theorem of Hunt \cite{Hunt} then implies that the generator must be
a sum of the three elementary forms $\mathscr{L}_H^{\mathrm{det}} , \mathscr{L}_R^{\mathrm{diff}} ,
\mathscr{L}_{S, \nu}^{\mathrm{jump}}$ given in (\ref{Ldet})-(\ref{Ljump}). This result provides the
complete characterization of the possible generators for QMSs which, in their
terminology, correspond to  ``essentially classical noise''.

\medskip

In the language of quantum feedback networks \cite{GJ09}, we may concatenate $SLH$ models - that is, run them in parallel - by making use of the concatenation product defined in (\ref{concat}). In the single input case ($d=1$), we have $U_H^{\mathrm{det}}$ determined simply by $G_H^{\mathrm{det}} \sim (1,0,H)$, while for $U_R^{\mathrm{diff}}$ it suffices to take $G^{\mathrm{diff}}_{R, \theta} \sim (\Eye, e^{i \theta} R, 0 )$, where $R=R^\ast$ and $\theta \in {\mathbb R}$ is some phase. The role of this phase is to determine which quadrature process to identify as the Wiener process: this should be
\begin{eqnarray}
W(t) = ie^{i \theta} B(t)^\ast - ie^{-i\theta } B(t).
\end{eqnarray}
Finally, to obtain a jump process $U_{S,\nu}^{\mathrm{jump}}$, it is enough to notice that
for any $\xi \in \mathbb{C}$,
\begin{eqnarray}
N_\xi (t) =\Lambda (t)
+ \xi B(t)^\ast + \xi^\ast B(t) + \nu t
\end{eqnarray}
is a Poisson process with rate $\nu = | \xi |^2>0$ for the vacuum state. Thus, for instance, 
$$G^{\mathrm{jump}}_{S, \xi } \sim \Big( S-\Eye, \xi (S-\Eye),\frac{| \xi |^2}{2i} (S^\ast - S) \Big),$$
would lead to an equivalent stochastic unitary. In this case we have exactly a model of the form (\ref{eq:U_jump}) with $N_\xi$ in the Fock vacuum state being identified as the Poisson process $N_\nu$. (A neater approach would be to identify $N_\nu$ with the number process $\Lambda$ in the coherent state with constant intensity $\xi$ over the time period of interest \cite{GJ09}.) 

\medskip

We now restate the K\"ummerer-Maassen theorem in SLH language: 

\begin{theorem}[K\"{u}mmerer-Maassen \cite{KM87}]
\label{thm:KM}
A QMS with essentially classical noise will be a concatenation of the single-input models as follows:
\begin{eqnarray}
G^{\mathrm{classical}} \sim   \Big(\boxplus_j G^{\mathrm{jump}}_{S_j , \xi_j} \Big) \boxplus
\Big(  \boxplus_k G^{\mathrm{diff}}_{R_k, \theta_k } \Big)\boxplus G_H^{\mathrm{det}} ,
\label{kmt}
\end{eqnarray}
where the $S_j$ are unitary operators, the $R_k$ and $H$ are self-adjoint operators, the complex numbers $\xi_j$ determine the Poisson rates ($ \nu_j = | \xi_j |^2 $) and the $\theta_k$ are phases. (The phases of the $\xi_j$ and the phases $\theta_k$ make no contribution to the generator.)
\end{theorem}

\medskip

\noindent
The original version of the Theorem gives the generators of the essentially classical QMSs to have the form (with the same notation)
\begin{eqnarray}
\mathscr{L}= \sum_j \nu_j \mathscr{L}_{S_j}^{\mathrm{jump}} + \sum_k \mathscr{L}_{R_k}^{\mathrm{diff}}
+\mathscr{L}_{H}^{\mathrm{det}} .
\end{eqnarray}
Additionally, they show that this is equivalent to the generators belonging to the closure of the cone generated by the jump generators.

We readily see that there exist dilations of maximally dephasing QMSs that are {\em not} essentially classical. Indeed, suppose that
$(\mathbf{S} , \mathbf{L} , H)$ leads to a maximal dephasing QMS with stable basis $\{| n \rangle \}$. Theorem \ref{thm:dephasing} only constrains the operators $\mathbf{L}$ and $H$: that is, they must take the forms (\ref{eq:L_max_decomp}) and (\ref{eq:H_max_decomp}), respectively. The freedom to choose $\mathbf{S}$ means that we may always perturb an essentially classical maximally dephasing model corresponding to $G\sim ( \Eye, \mathbf{L} , H)$ to get a genuine non-commutative one, $G^\prime \sim (\mathbf{S} , \mathbf{L} , H)$, with a matrix $\mathbf{S}$ that is no longer a multiple of the identity operator on $\mathbb{C}^N \otimes \mathbb{C}^d$.  The addition of $\mathbf{S}$ entails adding terms involving the processes $\Lambda_{jk}(t)$ to the QSDE \eqref{eq:U_QSDE}. When the bosonic fields (as the environment) are initialized in the vacuum state, the QSDE for the essentially commutative and non-commutative dilations $G$ and $G'$ described above yields the same QMS (which is maximally dephasing for appropriate choices of $\mathbf{L}$ and $H$, {\em independently} of the choice of $\mathbf{S}$). They also produce an identical evolution of the joint state of the system and fields. However, for other initial states of the fields for which the solution of the QSDE is well-defined (e.g., in the linear span of the coherent states of the fields \cite{HP84,Partha92}), they will {\em not}, in general, yield the same joint state evolution due to the presence of the terms $L^*_{j}S_{jk}dB_k(t)$ and terms involving the processes $\Lambda_{jk}(t)$. This is illustrated in the next example for a QMS with a single decoherence channel:  

\begin{example}
Consider a single qubit with  operators $H$ and $L \neq 0 $ diagonal in some fixed orthogonal basis. The associated QMS is therefore maximally dephasing. Consider a dilation of the qubit described by the following QSDE: 
$$ dU_G(t) = \!\bigg[ \!-\!\Big( \! i H+\frac{1}{2} L^*L\Big)\,dt + dB^*(t)L -L^*S \,dB(t) + (S-I)\,d\Lambda(t)\bigg] U_G(t), $$
where $S$ is unitary and different from the qubit identity operator. The QSDE is a dilation of the QMS that is {\em not} essentially commutative since the term $(S-I)d\Lambda(t)$ does not commute with the terms  $dB^*(t)L$ and  $-L^*S dB(t)$. It can also be seen that the generator on the right hand-side of the QSDE cannot be expressed in terms of processes that are commuting with themselves and one another for any two times $s,t \geq 0$ (i.e., they are not essentially classical processes). Nonetheless, since $H$ and $L$  are diagonal, the QMS obtained from the above QSDE, by tracing out the bosonic fields in a vacuum state, will be maximally dephasing, as stated.

It is worth recalling that a maximally dephasing QMS can be recovered from a non-classical  dilation in which the field state is not vacuum, but this does not hold generally. As an example, consider fixed-amplitude coherent states $|f\rangle_{\rm coh}$ of the field, with $f$ being a non-zero constant function of time. The dilation $G=(S,L,H)$ with the field initialized in the state $|f\rangle_{\rm coh}$ is then equivalent to a dilation $G'=(S,L+Sf,H+{\rm Im}\{L^{\dag}Sf\})$ with the field initialized in the vacuum. However, $\mathscr{L}_G \neq \mathscr{L}_{G'}$; in particular, $\mathscr{L}_{G'}$ will have a dependence on $S$, whereas $\mathscr{L}_G$ does not. If $L$ and $H$ satisfy the conditions of Theorem \ref{thm:dephasing}, $S$ commutes with $L$ and $H$, and  $L+Sf$ also satisfy the condition of  Theorem \ref{thm:dephasing}, then one may show that $\mathscr{L}_{G'}$ will also be the Lindblad generator of a maximally dephasing QMS. In general, however, one does not obtain a maximally dephasing QMS for an arbitrary initialization of the field state, including for Gaussian states of the field beyond coherent states, see e.g. the second reference in [16]).  In fact, if the field is initialized in states such as single-photon or multi-photon states, the reduced evolution on the system can no longer be described by a QMS; see \cite{GJN11} and the references therein.
\end{example}

\begin{theorem}
\label{thm:ess_class}
Let $\Phi _{t}$ be a QMS that is both maximally dephasing with respect to a
stable basis $\left\{ |n\rangle :n=1,\ldots ,N\right\} $ and
essentially classical. Then
\begin{eqnarray}
\Phi _{t}\left( |n\rangle \langle m|\right) =C_{nm}(t ) \,|n\rangle \langle m| ,
\label{eq:hadamard_continuous}
\end{eqnarray}
where the coefficients take the form $C_{nm}\left( t\right) = e^{ z_{nm} t}$, and
\begin{eqnarray}
z_{nm} = \sum_{j\in J}\nu _{j}\!\left( e^{-i\left( \vartheta
_{j,n}-\vartheta _{j,m}\right) }-1\right)  -\frac{1}{2}\sum_{k\in
K}(r_{k,n}-r_{k,m})^{2} +i\left( \varepsilon _{n}-\varepsilon
_{m}\right) ,
\label{eq:C_t}
\end{eqnarray}
where the $\nu _{j}$ are positive and the parameters $\vartheta_{j,n},r_{k,n},\varepsilon _{n}$ are real.
\end{theorem}
\begin{proof}
By the K\"{u}mmerer-Maassen theorem (\ref{kmt}), there exist two non-overlapping
subsets $J$ and $K$ whose union is $\left\{ 1,\cdots ,d\right\} $, such that
\begin{eqnarray}
dG^{\text{classical}}(t)  = &&
\!\sum_{j\in J}\left( S_{j}-1\right)
dN_{j}( t) -i\sum_{k\in K}R_{k}dW_{k}( t) \nonumber \\
 && - \Big( \frac{1}{2}\sum_{k\in K}R_{k}^{2}+iH\!\Big) dt ,
\label{Gcl}
\end{eqnarray}
with $N_{j}=\Lambda _{jj}+\xi _{j}B_{j}^{\ast }+\xi _{j}^{\ast }B_{j}+\left| \xi _{j}\right| ^{2}$ and $W_{k}=ie^{i\theta _{k}}B_{k}^{\ast }-ie^{-i\theta_{k}}B_{k}$. The quantum processes $\left\{ N_{j},W_{k}:j\in J,k\in K\right\} $ form a commuting set of self-adjoint processes. Thus, we may decompose the multiplicity space as $\mathfrak{K} = \mathfrak{K}^{\text{jump}}\oplus \mathfrak{K}^{\text{diff}} ,$ where $\mathfrak{K}^{\text{jump}}=\mathbb{C}^{J}$ and $\mathfrak{K}^{\text{diff}}=\mathbb{C}^{K}$.

Associated with the jumps and diffusion terms we have the respective coupling
operators $L_{j}^{\text{jump}}=\xi _{j}\left( S_{j}-1\right) $, for each
$j\in J$ and $L_{k}^{\text{diff}}=e^{i\theta _{k}}R_{k}$, for each $k\in K$.
If the QMS\ is to be maximally dephasing, these must be diagonal in the
stable basis, so that
\begin{eqnarray*}
S_{j}\equiv \sum_{n=1}^{N}e^{i\vartheta _{j,n}}|n\rangle \langle
n|,\qquad R_{k}\equiv \sum_{n=1}^{N}r_{k,n}|n\rangle \langle n|.
\end{eqnarray*}
(Likewise, $H\equiv \sum_{n=1}^{N}\varepsilon _{n}|n\rangle \langle
n|.$) Moreover, the set of vectors
\begin{eqnarray*}
| \boldsymbol{\mathbf{\lambda }}_{n}\rangle =\left[
\begin{array}{c}
\left[ \xi _{j}(e^{i\vartheta _{j,n}}-1)\right] _{j\in J} \\
\left[ e^{i\theta _{k}}r_{k,n}\right] _{k\in K}
\end{array}
\right]
\end{eqnarray*}
are linearly independent in $\mathfrak{K}=\mathfrak{K}^{\text{jump}}\oplus \mathfrak{K}^{\text{diff}}$. We
may therefore write
\begin{eqnarray*}
dG^{\text{classical}}( t) =\sum_{n=1}^{N}|n\rangle
\langle n|\otimes d\tilde{G}_{n}\left( t\right),
\end{eqnarray*}
in terms of a family of processes 
\begin{eqnarray*}
\tilde{G}_{n}( t)  =
\sum_{j\in J}\left( e^{i\vartheta_{j,n}}-1\right) N_{j}\left( t\right) -i\sum_{k}r_{k,n}W_{k}\left( t\right)
-\Big( \frac{1}{2}\sum_{k\in K}r_{k,n}^{2}+i\varepsilon _{n}\!\Big) t.
\end{eqnarray*}
Since we have $\left( dN_{j}\right) ^{2}=dN_{j}$ and $\left( dW_{k}\right)^{2}=dt$,
with all other products of increments vanishing, we may integrate to get
\begin{eqnarray}
U\left( t\right) \equiv \sum_{n=1}^{N}|n\rangle \langle n|\otimes e^{-i\Theta _{n}\left( t\right) } ,
\label{eq:U_diagonal}
\end{eqnarray}
where, for each $n$, 
\begin{eqnarray*}
\Theta _{n}\left( t\right) =\sum_{j\in J}\vartheta _{j,n}N_{j}\left(
t\right) +\sum_{k}r_{k,n}W_{k}\left( t\right) +\varepsilon _{n}.
\end{eqnarray*}
The processes $\{ \Theta_n \}$ all commute. It follows that the QMS obtained in this way is of the form (\ref{eq:hadamard_continuous}), 
with coefficients $C_{nm}\left( t\right) =\left\langle \Omega ,e^{i\left( \Theta _{n}(t) -\Theta _{m}( t) \right) }\,\Omega \right\rangle ,$ with $\Omega$ being, as before, the Fock vacuum. Noting that the $N_{j}$ and $W_{k}$ obey the laws of independent Poisson and Wiener processes in this state, we are lead to the desired expression (\ref{eq:C_t}).
\end{proof}

\smallskip

We see that the essentially commutative dilations of a maximally dephasing QMS have the feature that  they are diagonal in the 
stable basis, Eq.  (\ref{eq:U_diagonal}). As pointed out earlier on, however, we may perturb this with an additional matrix $\mathbf{S}$ which need not respect the stable basis so as to end up with a non-commutative dilation which still retains the maximal dephasing property.

\subsection{Essentially classical dilation by a diffusion or a jump process?}

We note that both the diffusion and jump QMS in (\ref{Ldiff}) and (\ref{Ljump}) may be expressed as 
\begin{eqnarray}
\mathscr{L}_c (X) \equiv \frac{1}{2} [ c^\ast , X ]c + \frac{1}{2} c^\ast [X,c ] ,
\label{eq:L_c}
\end{eqnarray}
with $c= \lambda \, R$, $\lambda \in \mathbb{C}$, $R$ self-adjoint,  and $c= \sqrt{\nu} \, S$, $\nu >0$, $S$ unitary, respectively.
(The transformation $c \mapsto e^{i \varphi} c+ \beta$, for $\varphi \in \mathbb{R}$ and $\beta \in \mathbb{C}$, 
leaves $\mathscr{L}_c$ unchanged).

K\"{u}mmerer and Maassen also showed, in Proposition 2.2.1 of \cite{KM87}, that a QMS with generator of the form $\mathscr{L}_c$ in (\ref{eq:L_c}) is essentially classical if and only if $c$ is normal and has a spectrum which lies either on a straight line or on a circle in the complex plane. As a corollary, a generator of the form (\ref{eq:L_c}), with $c$ normal, may be {\em both} a diffusion type or a jump type if and only if the spectrum of $c$ consists of no more than two points (since it lies in the intersection of a line and a circle).

For instance, in the paradigmatic example of pure dephasing of a qubit discussed in \S \ref{pd}, the Lindbladian (\ref{eq:qubit_dephase_Lindbladian}) can be considered as either a diffusion type (with $ \gamma = |\lambda |^2 , \, 
R = \sigma_z$) or as a jump type (with $ \nu = \gamma ,\, S= \sigma_z$). Therefore, the QMS for pure dephasing of 
a qubit can arise as the average of {\em either} a diffusive model or jump model. This is fortuitous, as the operators $R$ 
and $S$ can have only two eigenvalues each in the qubit case.

In order to see that the above scenario is far from generic, 
let us first still assume $N=2$, but consider a simple modification to the above pure-dephasing example, where
we take the scattering matrix to be 
\begin{eqnarray*}
S =
\bigg[
\begin{array}{cc}
	1 &0  \\
	0 & e^{i \vartheta}
\end{array}
\bigg], \quad e^{i \vartheta} \neq \pm 1.
\end{eqnarray*}
The generator $\mathscr{L}^{\mathrm{jump}}_{\mathbf{S}, \nu}( X)= \nu (S^\ast XS - X)$ then becomes
\begin{eqnarray*}
\mathscr{L}^{\mathrm{jump}}_{\mathbf{S}, \nu} \bigg(
\bigg[
\begin{array}{cc}
	x &y \\
	z & w
\end{array}
\bigg]
\bigg)
= \nu
\bigg[
\begin{array}{cc}
	0 & (e^{i \vartheta} -1) y \\
	(e^{-i \vartheta}-1) z & 0
\end{array}
\bigg] ,
\end{eqnarray*}
and the solution to the master equation reads
\begin{eqnarray*}
\rho (t) =
\bigg[\!
\begin{array}{cc}
	\rho_{11} (0) & e^{\nu (e^{-i \vartheta} -1) t} \rho_{10} (0)  \\
	e^{\nu  (e^{i \vartheta} -1) t} \rho_{01} (0)  & \rho_{00} (0)
\end{array}\!
\bigg].
\label{eq:dephasing_theta}
\end{eqnarray*}
In fact, the only difference compared to the usual pure-dephasing (\ref{eq:pure_dephase_rho}) is that the 
damping constant is now complex. Its real part, $\nu (\cos \vartheta -1)$, is strictly negative since
$e^{-i \vartheta}=1$ is excluded. Therefore, we once again have dephasing.

To see what is going on, observe that the corresponding unitary stochastic process has the explicit
solution determined by (\ref{eq:U_jump}), namely,
\begin{eqnarray*}
U_{S,\nu}^{\mathrm{jump}} (t) =
\bigg[
\begin{array}{cc}
	1  & 0 \\
	0 & e^{i \vartheta N_\nu (t)}
\end{array}
\bigg],
\end{eqnarray*}
with the result that
\begin{eqnarray*}
U_{S,\nu}^{\mathrm{jump}} (t)^\ast \bigg[
\begin{array}{cc}
	x &y \\
	z & w
\end{array}
\bigg]
U_{S,\nu}^{\mathrm{jump}} (t)
=\bigg[
\begin{array}{cc}
	x & e^{i \vartheta N_\nu (t)} y \\
	e^{-i \vartheta N_\nu (t)}z & w
\end{array}
\bigg].
\end{eqnarray*}
Accordingly, the dephasing in the long-time limit can be seen to be due to the random phase accumulation on off-diagonal elements. We note that the generator is bistochastic but, for $e^{i \vartheta} \neq \pm 1$, it is not self-dual (recall the discussion in \S \ref{sec:bistochastic}).

\medskip

In higher dimension, $N>2$, we may diagonalize an arbitrary unitary $\mathbf{S}$ in the general form $ \sum_n e^{-i \vartheta_n} \, |n \rangle \langle n |$, so that
\begin{eqnarray*}
\mathscr{L}_{\mathbf{S}, \nu}^{\mathrm{jump} } ( |n \rangle \langle m |) =   
- \left( 1-   e^{ i ( \vartheta_n - \vartheta_m) }\right)  |n \rangle \langle m | .
\label{eq:SS}
\end{eqnarray*}
This time, we have maximally dephasing behavior with respect to the basis provided by $\vartheta_n \neq \vartheta_m$ for $n \neq m$.
This will be self-adjoint only in the very restrictive situations discussed above. Note that $\mathscr{L}_{S, \nu }^{\mathrm{jump} }  (|n \rangle \langle n |) =0$, so the QMS is, consistently, dephasing with respect to the basis $\{ \vert n \rangle \}$.

Evidently, in the situation of classical noise leading to dephasing, the generator does {\em not} need to be self-dual.

\subsection{Classical dilations via diffusions}

To realize a maximally dephasing QMS through a diffusive dilation, take $G\sim \left( \mathbf{S}=\Eye,\mathbf{L},H\right) $,
where the coupling operators are of the form $L_{k}\equiv \sum_{n}\lambda_{k,n}\,|n\rangle \langle n|$ and the Hamiltonian
is $H\equiv \sum_{n}\varepsilon _{n}\,|n\rangle \langle n|$. The unitary has now the differential germ
\begin{eqnarray}
dG\left( t\right) &=&\sum_{k}\left( L_{k}\otimes dB_{k}\left( t\right)
^{\ast }-L_{k}^{\ast }\otimes dB_{k}\left( t\right) +K\otimes dt\right) \label{dG} \nonumber \\
&\triangleq  &\sum_{n}|n\rangle \langle n|\otimes \left\{
-idQ_{n}\left( t\right) +\kappa _{n}dt\right\} ,
\end{eqnarray}
where we have introduced the processes
\begin{eqnarray*}
Q_{n}\left( t\right) =i\sum_{k}\left\{ \lambda _{k,n}B_{k}\left( t\right)
^{\ast }-\lambda _{k,n}^{\ast }B_{k}\left( t\right) \right\} .
\end{eqnarray*}
Each of the processes $\left( Q_{n}\left( t\right) \right) _{t\geq 0}$ has 
the statistics of a Wiener process for the Fock vacuum state: $(dQ_n )^2 = \sigma_n ^2 dt$, with $\sigma_n^2 =
\sum_k | \lambda_{k,n} |^2$. However, they may not be compatible. In fact, we readily see that
\begin{eqnarray}
\left[ Q_{n}( t) ,Q_{m}( s) \right] =2i \, A_{nm} \min \left( t,s\right) ,
\end{eqnarray}
where $A_{nm}$ is the symplectic area defined in (\ref{eq:symplectic_area}).

Note that the Stratonovich form of the It\={o} QSDE (\ref{dG}) is
\begin{eqnarray*}
dU\left( t\right) =-i\bigg[
\sum_{n}|n\rangle \langle n|\otimes \Big( dQ_{n}(t) +\varepsilon_{n}dt\Big) \bigg] \circ dU( t) .
\end{eqnarray*}
\noindent 
Together with (\ref{Gcl}), this leads to the following result for quantum diffusions: 

\smallskip

\begin{theorem}
\label{thm:dephasing_diffusion}
Consider a maximally dephasing QMS, with stable basis $\left\{
|n\rangle :n\right\} $, represented by $G\sim \left(\mathbf{S}=\Eye,\mathbf{L},H\right) $, 
with coupling operators $L_{k}\equiv \sum_{n}\lambda_{k,n}\,|n\rangle \langle n|$ and Hamiltonian 
$H\equiv \sum_n \varepsilon_n |n\rangle\langle n|$. Then the QMS admits an essentially 
classical diffusive dilation
if and only if the symplectic areas 
$A_{nm} =\mathrm{Im} \langle \boldsymbol{\mathbf{\lambda}}_{n}| \boldsymbol{ \mathbf{\lambda}} _{m} \rangle $ vanish for each pair $n,m$. 
In this case we have
\begin{eqnarray}
U(t) = \sum_n | n \rangle \langle n | \otimes e^{- \frac{1}{2} \sigma_n^2 t} \, e^{-i \left( Q_n (t) + \varepsilon _n t \right) },
\end{eqnarray}
where the $Q_n$ are independent, compatible quantum Wiener processes with variances
$\sigma_n^2 = \langle \boldsymbol{\mathbf{\lambda}}_n | \boldsymbol{\mathbf{\lambda}}_n \rangle = \sum_k | \lambda_{k,n} |^2$.
\end{theorem}

\subsection{Classical dilations via jumps}   The above theorem implies that vanishing of the 
obstruction serves as a witness to the classicality of an underlying diffusive dilation. The situation is 
more subtle if general classical dilations including Poisson processes are allowed. Recall the expression for
the coefficients $z_{nm}$ occurring in (\ref{eq:C_t}). These yield the dephasing rates, see (\ref{eq:dephase_rates}),
\begin{eqnarray}
\gamma_{nm} =\frac{1}{2} \sum_{j\in J}\nu _{j} |  e^{i \vartheta
_{j,n}}-e^{i\vartheta _{j,m}} |^2 +\frac{1}{2}\sum_{k\in
K} (r_{k,n}-r_{k,m})^{2} ,
\end{eqnarray}
and we obtain a similar expression for the dephasing frequencies $\omega_{nm}$, see (\ref{eq:dephase_frequency}).
We note that the obstruction will in this case be given by
\begin{eqnarray}
\label{eq:obstruction-phases}
\Delta_{nml} =  \sum_{j\in J}\nu _{j} \Big\{
 \sin( \vartheta_{j,m}-\vartheta _{j,n})+
 \sin( \vartheta_{j,l}-\vartheta _{j,m})+
 \sin( \vartheta_{j,n}-\vartheta _{j,l})
\Big\} .
\end{eqnarray}
Thus, the obstruction may be non-zero and this is entirely down to the Poissonian terms. 

\begin{remark}
By combining Theorem \ref{thm:dephasing_diffusion} with the above result, it follows that
(i) if a QMS is maximally dephasing, then vanishing of the Hamiltonian obstruction is {\em sufficient} 
(but {\em not} necessary) for an essentially classical dilation to exist;
(ii) if a QMS is maximally dephasing and essentially classical, a non-zero obstruction 
can {\em only} arise due to the presence of Poissonian noise.
\end{remark}

\section{Conclusion}

We have developed a notion of dephasing under the action of a QMS in terms of the convergence of operators to a block-diagonal form corresponding to irreducible invariant subspaces. An important special case is maximal dephasing, occurring when all such invariant subspaces are mutually orthogonal and one-dimensional. Our notion includes the definition of dephasing relative to a preferred energy basis, as formalized in \cite{AFGG12,AFG13}. 

We have characterized the maximal dephasing setting in detail, obtaining, in particular, the maximal rank for a maximally dephasing QMS. We further show that the phase component in the decay terms for off-diagonal elements of an operator evolving under such a QMS (that is, the imaginary parts of the generator's eigenvalues) need not come from a Hamiltonian, which we refer to as a Hamiltonian obstruction. QMSs which are maximally dephasing and free of obstruction are precisely those for which a representation of the generator solely in terms of self-adjoint operators exists.

A main motivating question for this work has been determining whether the evolution under a dephasing QMS may result from a dilation to a unitary stochastic dynamics with classical commutative noise, namely, an essentially classical dilation in K\"{u}mmerer and Maassen's terminology \cite{KM87}.  We have taken steps toward answering this question by employing the results developed for maximally dephasing QMSs to study their dilations by classical noise. We show that, remarkably, a diffusive dilation of such a QMS can occur {\em if and only if there is no Hamiltonian obstruction}. From this, we further establish that if a maximally dephasing QMS has non-zero obstruction and admits a classical dilation, then the obstruction can {\em only} arise from classical Poisson processes.

As a result of independent interest, we also show that any maximally dephasing QMS always admits a genuinely non-classical unitary dilation that possesses the same rank-one invariant subspaces. 

The present analysis leaves a number of open questions for future investigation.  Most importantly, it would be desirable, both conceptually and practically, to find stronger criteria that may be able to diagnose the existence of a \emph{general} classical unitary dilation - including both diffusive {\em and} Poisson processes - directly from the structure of the underlying dephasing QMS generator.  This analysis would also extend, to the case of continuous-time Markov dynamics, existing investigations on the extent to which discrete-time dephasing evolutions (phase-damping channels) may be represented in terms of classical random unitary dynamics \cite{Strunz}. In a similar venue, it is natural to ask whether a non-vanishing Hamiltonian obstruction (or a possible stronger measure) could be brought to bear on other non-classical dynamical traits - such as loss of positivity in appropriate quasi-probability distribution functions \cite{Nori}, generation of entanglement between the system and the environment \cite{Cywinski} or some observed fraction thereof, as relevant to understand the emergence of redundant information encoding through spectrum broadcast structures \cite{Horodecki}.  Lastly, extensions beyond the maximal dephasing setting (and, ultimately, beyond pure dephasing dynamics) are worth investigating in the light of connections with more general information-preserving structures \cite{IPS}. We plan to report on some of these issues in a separate study \cite{Next}.

\section*{Acknowledgements}
All authors wish to thank the Institut Henri Poincar\'{e} for the kind hospitality and support during 
the IHP Trimester \textit{Measurement and Control of Quantum Systems: Theory and Experiments}, 
Paris, France, 2018. HN gratefully acknowledges support from the Australian Research Council through 
Discovery Project DP130104191. LV also gratefully acknowledges partial support by the US National 
Science Foundation under Grant No. PHY-1620541.

\end{document}